\documentclass[a4paper,10pt]{amsart}
\usepackage[utf8]{inputenc}
\usepackage[T1]{fontenc}
\usepackage{amsmath}
\usepackage{amsthm}
\usepackage{mathtools}
\usepackage{amsfonts}
\usepackage{amssymb}
\usepackage{mathrsfs}
\usepackage{color}
\usepackage{graphicx}
\usepackage{enumitem}
\usepackage{url}

\newcommand{\C}{\mathbb{C}}
\newcommand{\R}{\mathbb{R}}

\DeclareMathOperator{\sgn}{\mathrm{sgn}}
\newcommand{\dd}{\mathrm{d}}
\newcommand{\ee}{\mathrm{e}} 

\newcommand{\dom}{\mathrm{Dom}}
\newcommand{\ran}{\mathrm{Ran}}
\newcommand{\res}{\mathrm{res}}
\newcommand{\tr}{\mathrm{Tr}}

\newtheorem{theorem}{Theorem}[section]
\newtheorem{lemma}[theorem]{Lemma}
\newtheorem{corollary}[theorem]{Corollary}
\newtheorem{proposition}[theorem]{Proposition}

\theoremstyle{definition}
\newtheorem{remark}[theorem]{Remark}

\title[Non-self-adjoint relativistic point interaction]{Non-self-adjoint relativistic point interaction \\in one dimension}
\author{Luk\'{a}\v{s} Heriban}
\address{Department of Mathematics\\ Faculty of Nuclear Sciences and Physical Engineering \\ Czech Technical University in Prague\\ Trojanova 13\\ 120 00 Prague\\ Czechia}
\email{heribluk@fjfi.cvut.cz}
\author{Mat\v{e}j Tu\v{s}ek}
\address{Department of Mathematics\\ Faculty of Nuclear Sciences and Physical Engineering \\ Czech Technical University in Prague\\ Trojanova 13\\ 120 00 Prague\\ Czechia}
\email{matej.tusek@fjfi.cvut.cz}
\date{May 9, 2022}

\begin{document}

\begin{abstract}
The one-dimensional Dirac operator with a singular interaction term which is formally given by $A\otimes|\delta_0\rangle\langle\delta_0|$, where $A$ is an arbitrary $2\times 2$ matrix and $\delta_0$ stands for the Dirac distribution, is introduced as a  closed not necessarily self-adjoint operator. We study its spectral properties, find its non-relativistic limit and also address the question of regular approximations. In particular, we show that, contrary to the case of local  approximations, for non-local approximating potentials, coupling constants are not renormalized in the limit.
\end{abstract}

\maketitle

\section{Introduction}
Consider  formal differential expressions
\begin{equation} \label{eq:formal}
\mathscr{D}:=\sigma_1\otimes(-i\dd/\dd x)+\sigma_3\otimes m,\quad \mathscr{D}_A:=\mathscr{D}+A\otimes|\delta_0\rangle\langle\delta_0|,\end{equation}
where 
\begin{equation*}
\sigma_1=\begin{pmatrix}
0 & 1\\
1 & 0
\end{pmatrix},
\quad
\sigma_3=\begin{pmatrix}
1 & 0\\
0 & -1
\end{pmatrix}
\end{equation*}
are the Pauli matrices, $m\in\R$,  $A$ is an arbitrary  $2\times 2$ complex matrix, and $\delta_0$ is the Dirac distribution supported at $x=0$.  We would like to introduce $\mathscr{D}_A$ as a well defined closed operator in $L^2(\R)\otimes\C^2\equiv L^2(\R;\C^2)$. To this aim,  we   firstly extend the action of the distribution $\delta_0$ to functions which are not necessarily continuous at $x=0$.
Following \cite{Ku_96} and \cite[Sect. 3.2.4]{singular}, we put
\begin{equation*}
\langle\delta_0,\psi\rangle:=\frac{\psi(0_+)+\psi(0_-)}{2}\in\C^2
\end{equation*}
for every $\psi\equiv\psi_-\oplus\psi_+\in H^1(\R_-;\C^2)\oplus H^1(\R_+;\C^2)\subset L^2(\R_-;\C^2)\oplus L^2(\R_+;\C^2)\equiv L^2(\R;\C^2)$, where  by  $\psi(0_-)$ and $\psi(0_+)$ we mean the values of the continuous representatives of $\psi_-$ and $\psi_+$, respectively, i.e., the one-sided traces at $x=0$. For such a $\psi$, we have
$$\mathscr{D}\psi=\mathscr{D}\psi_-\oplus \mathscr{D}\psi_+ -i \sigma_1(\psi(0+)-\psi(0_-))\delta_0.$$
If we want $\mathscr{D}_A\psi$ to remain in $L^2(\R;\C^2)$, the singular contributions have to cancel out. This yields
$$-i\sigma_1 (\psi(0+)-\psi(0_-))+A~\frac{\psi(0_+)+\psi(0_-)}{2}=0,$$ 
which is convenient to rewrite as
\begin{equation} \label{eq:TC}
(2i\sigma_1-A)\psi(0_+)=(2i\sigma_1+A)\psi(0_-),
\end{equation}
which motivates us to define the following operator
\begin{align*}
&\dom(D_A)=\{\psi\equiv\psi_-\oplus\psi_+\in H^1(\R_-;\C^2)\oplus H^1(\R_+;\C^2)|\, \eqref{eq:TC}\text{ holds}\}\\
& D_A\psi=\mathscr{D}\psi_-\oplus \mathscr{D}\psi_+.
\end{align*}

Note that, for every matrix $A$, $D_A$ is an extension of the symmetric operator $D_{\min}$ that acts as $\mathscr{D}$ on 
$$\dom(D_{\min}):=\{\psi\in H^1(\R;\C^2)|\, \psi(0)=0\}.$$
It is a  well known fact that $(D_{\min})^*=D_{\max}$, where
\begin{align*}
&\dom(D_{\max})=\{\psi\equiv\psi_-\oplus\psi_+\in H^1(\R_-;\C^2)\oplus H^1(\R_+;\C^2)\}\\
& D_{\max}\psi=\mathscr{D}\psi_-\oplus \mathscr{D}\psi_+,
\end{align*}
cf. \cite{PaRi_14}. Therefore, $D_A$ is a restriction of $D_{\max}$. The deficiency indices of $D_{\min}$ are $(2,2)$ \cite{BeDa_94,PaRi_14}. Consequently, there is a four-real-parametric family of self-adjoint extensions of $D_{\min}$. We will show below that most of these self-adjoint extensions correspond to operators $D_A$ with hermitian matrices $A$ and that the remaining self-adjoint extensions cannot be described by the transmission condition \eqref{eq:TC}. If $(2i\sigma_1-A)$ is invertible then \eqref{eq:TC} yields 
\begin{equation} \label{eq:TCLambda}
\psi(0_+)=\Lambda\psi(0_-),
\end{equation}
where $\Lambda$ is a multiple of a matrix with real diagonal terms and purely imaginary off-diagonal terms such that $|\det(\Lambda)|=1$. The self-adjoint extensions described by the transmission condition \eqref{eq:TCLambda} were thoroughly studied in \cite{BeDa_94}. Two important special cases, namely the relativistic counterparts of the non-relativistic $\delta$- and $\delta'$- interactions, were investigated even earlier in \cite{GeSe_87} (see \cite{CaMaPo_13} for a generalization to infinitely many interaction points). Note that the transmission condition equivalent to \eqref{eq:TCLambda} appeared also in \cite{DiExSe_89}. Using the framework of boundary triplets it is possible to analyze all self-adjoint extensions of $D_{\min}$ in a uniform manner \cite{PaRi_14}.

The main aim of the present paper is to study $D_A$ with a general matrix $A$. Most of the new results will concern the case when $A$ is not hermitian, which is, as far as we know, the setting that has not been considered before. Firstly, in Section \ref{sec:basic}, we will check that $D_A$ is a closed operator, we will find its adjoint and show that $D_A$ is self-adjoint if and only if $A$ is hermitian, as one would guess from the formal expression \eqref{eq:formal}. Furthermore, we will provide a full description of matrices $A$ for which $D_A$ decouples into a direct sum of a pair of operators acting on the half-axes. Note that the self-adjoint realizations of the Dirac operator on a half-axis were studied before in \cite{GrLe_06}.

Section \ref{sec:spectrum} is devoted to the spectral analysis of $D_A$. Essentially, it extends  results of \cite{BeDa_94} to a not necessarily self-adjoint setting.  The spectrum of the free operator $D_0$ is purely absolutely continuous and equals $(-\infty,-|m|]\cup[|m|,+\infty)$. When  a self-adjoint point interaction is switched on, this part of the spectrum is preserved and at most two (counting multiplicities) eigenvalues may occur in the gap. In the non-self adjoint setting,  spectral effects of the point interaction may be much stronger. In fact, for certain choices of $A$ and $m$, the point spectrum of $D_A$ is the whole complex plane without $(-\infty,-|m|]\cup[|m|,+\infty)$ or either upper or lower open complex half-plane. However, except for these \textit{strongly non-self-adjoint} cases, there may be at most two (counting geometric multiplicities) discrete eigenvalues of $D_A$ in $\C\setminus((-\infty,-|m|]\cup[|m|,+\infty))$ and the essential spectrum of $D_A$ is preserved, similarly to the self-adjoint case. See Theorem \ref{theo:spec} and Corollary \ref{cor:spec} for details.  

In Section \ref{sec:approx}, we firstly study approximations of point interactions by more realistic regular potentials. The formal expression \eqref{eq:formal} for $D_A$ suggests  that $D_A^\varepsilon:=D_0+A\otimes|v_\varepsilon\rangle\langle v_\varepsilon|$, where $v_\varepsilon\to\delta_0$ as $\varepsilon\to 0$ in the distributional sense, may be a good candidate for  such a (non-local) potential. In fact, we will show that $D_A^\varepsilon$ converges to $D_A$ in the norm resolvent sense as $\varepsilon\to 0$, see Theorem \ref{theo:approx}. This result was obtained before in \cite{Se_89} for two special hermitian choices of $A$, which correspond to the electrostatic and the Lorentz scalar potential, although a complete proof was not presented there. In the same paper, the question of how to approximate relativistic point interactions (again for the same two special choices of $A$) by local regular potentials was addressed rigorously for the first time. Starting with $\mathscr{D}+A\otimes\delta_0$ instead of $\mathscr{D}_A$, where $\delta_0$ acts as a multiplication operator on functions that are not necessarily continuous at $x=0$, and so it has to be properly extended; one arrives at the same transmission condition as for $D_A$. Therefore, it is equally reasonable to take $D_0+A\otimes v_\varepsilon$ for approximating operators. Such a family converges in the norm resolvent sense as $\varepsilon\to 0$ but, surprisingly, not to $D_A$ but to the operator that describes the same type of interaction with a different coupling constant. The same effect was later observed in a general self-adjoint case \cite{Hu_97,Hu_99,Tu_20} and also in  higher dimensional settings for Dirac operators with $\delta$-shell interactions \cite{MaPi_17,MaPi_18,CaLoMaTu}. 
Since, on the other hand, the non-local approximations do not require any renormalization of the coupling constant, one is tempted to conclude that the very nature of the relativistic point interactions is non-local.
In the second part of the section, we will find an implicit equation for the eigenvalues of $D_A^\varepsilon$ and study their qualitative properties, see Propositions \ref{prop:ev} and \ref{prop:app_spec}.

Section \ref{sec:limit} deals with the non-relativistic limit of the operator $D_A^{m,c}$ that is derived from $D_A$ by adding the speed of light $c$ at right places, i.e., by changing physical units. We are then interested in $\lim_{c\to+\infty}(D_A^{m,c}-mc^2)$. One would guess that the limit operator (in a certain well defined sense) could be the non-relativistic Hamiltonian $H_A$ with $\delta$-interaction, which is described by the formal expression
$$\frac{1}{2m}\Big(-\frac{\dd^2}{\dd x^2}+\alpha|\delta_0\rangle\langle\delta_0|+i\beta|\delta_0\rangle\langle\delta'_0|-i\gamma|\delta'_0\rangle\langle\delta_0|+\delta|\delta'_0\rangle\langle\delta'_0|\Big)\equiv H_A.$$
Here $\alpha,\beta,\gamma,\delta\in\C$ are the elements of $A$,
\begin{equation} \label{eq:A}
A=\begin{pmatrix}
\alpha &\beta\\
\gamma &\delta
\end{pmatrix}.
\end{equation}
In the self-adjoint case, it was shown that some of the elements of $A$ have to be made $c$-dependent to get $1\leftrightarrow 1$ correspondence between relativistic operators $D_A$ and their non-relativistic counterparts \cite{BeDa_94}. Namely, written schematically, it holds
\begin{equation*}
"\lim_{c\to+\infty}\big(D_{A_c}^{m,c}-mc^2\big)=H_A",
\end{equation*}
where 
\begin{equation*}
A_c:=\begin{pmatrix}
\frac{1}{2mc}\alpha & \beta\\
\gamma & 2mc\delta
\end{pmatrix}.
\end{equation*}
We will prove that with the same scaling of $A$, this result may be extended to the non-self-adjoint setting, see Theorem \ref{theo:limit} for the precise statement.

Finally, some fundamental properties of $H_A$ are summarized in a separate appendix at the end of the paper. In particular, an explicit formula for the resolvent $(H_A-z)^{-1}$ is provided there.

\section{Basic properties} \label{sec:basic}
\begin{proposition} \label{prop:basic}
The operator $D_A$ is densely defined, closed, and uniquely determined by the matrix $A$ in the sense that $D_A=D_B$ if and only if $A=B$.
\end{proposition}
\begin{proof}
For every matrix $A$, $C_0^\infty(\R\setminus\{0\};\C^2)\subset\dom(D_A)$ and $C_0^\infty(\R\setminus\{0\};\C^2)$ is dense in $L^2(\R;\C^2)$. Hence, $D_A$ is densely defined. 

The norm $\|\psi\|_A^2:=\|\psi\|^2+\|D_A\psi\|^2$ defined for all $\psi\in\dom(D_A)\subset H^1(\R_-;\C^2)\oplus H^1(\R_+;\C^2)$ is equivalent to the norm in $H^1(\R_-;\C^2)\oplus H^1(\R_+;\C^2)$. Since the latter is complete, the limit $\psi$ of any Cauchy sequence $(\psi_n)\subset\dom(D_A)$ with respect to the norm  $\|.\|_A$ remains in $H^1(\R_-;\C^2)\oplus H^1(\R_+;\C^2)$. Moreover, $\psi$ obeys the same transmission condition \eqref{eq:TC} as all functions $\psi_n$, because the trace operator at $x=0$ is continuous mapping from $H^1(\R_\pm;\C^2)$ to $\C^2$. This shows that $\psi\in\dom(D_A)$, and so $A$ is closed.

Finally, $A=B$ clearly implies $D_A=D_B$. If $D_A=D_B$ then every $\psi\in\dom(D_A)=\dom(D_B)$ satisfies \eqref{eq:TC} for both $A$ and $B$. Subtracting these conditions we get 
\begin{equation} \label{eq:AB}
(A-B)(\psi(0_+)+\psi(0_-))=0.
\end{equation}
Take an arbitrary $\xi\in\C^2$ and construct a function $\psi\in H^1(\R_-;\C^2)\oplus H^1(\R_+;\C^2)$ such that $\psi(0_\pm)=\frac{1}{4}(2\sigma_0\mp i \sigma_1 A)\xi$.  Then $\psi\in\dom(D_A)$ and $(\psi(0_+)+\psi(0_-))=\xi$. Hence, the sum $(\psi(0_+)+\psi(0_-))$ in \eqref{eq:AB} attains all possible values. We infer that $A=B$.
\end{proof}

\begin{proposition} \label{prop:decoup}
The operator $D_A$ decouples into a direct sum of operators on $L^2(\R_-;\C^2)\oplus L^2(\R_+;\C^2)$ if and only if $A=\pm 2i\sigma_1$ or none of the matrices $(2i\sigma_1-A)$ and $(2i\sigma_1+A)$ is invertible. In the positive case, $D_A=D_A^+\oplus D_A^-$ with
\begin{equation} \label{eq:decoupOp}
\begin{split}
&\dom{D_A^\pm}=\{\psi_\pm\in H^1(\R_\pm;\C^2)|\, (2i\sigma_1\mp A)\psi_\pm(0)=0\} \\
&D_A^\pm\psi_\pm=\mathscr{D}\psi_\pm.
\end{split}
\end{equation}
\begin{proof}
If $A=2i\sigma_1$ then \eqref{eq:TC} yields $4i\sigma_1\psi(0_-)=0$ and no restriction on $\psi(0_+)$. The first condition is equivalent to $\psi(0_-)=0$. If $A=-2i\sigma_1$ then the roles of $\psi(0_-)$ and $\psi(0+)$ must be interchanged. In both cases, $D_A$ decouples in the manner that is asserted in the proposition. 
Next, we will assume that $A\neq\pm 2i\sigma_1$.

If $(2i\sigma_1-A)$ or  $(2i\sigma_1+A)$ is invertible then we rewrite \eqref{eq:TC} as $\psi(0_+)=\Lambda\psi(0_-)$ or $\psi(0_-)=\tilde\Lambda\psi(0_+)$, respectively. The range of $\Lambda$ and $\tilde\Lambda$, respectively, is at least one-dimensional. Therefore, the transmission condition \eqref{eq:TC} does not decouple.

Recalling \eqref{eq:A},
we conclude that both $(2i\sigma_1-A)$ and $(2i\sigma_1+A)$ are not invertible
if and only if
$$0=\det(2i\sigma_1\pm A)=\det(A)+4\mp 2i(\beta+\gamma),$$
which is equivalent to
\begin{equation} \label{eq:decoup_cond}
\beta+\gamma=0,\quad \det(A)=-4.
\end{equation}
Using these identities, it is straightforward to deduce that
\begin{equation} \label{eq:sigmaAsquared}
(\sigma_1 A)^2=-4\sigma_0. 
\end{equation}
Multiplying \eqref{eq:TC} by $\sigma_1$ we get
$$(2i\sigma_0-\sigma_1 A)\psi(0_+)=(2i\sigma_0+\sigma_1 A)\psi(0_-).$$
Applying $(2i\sigma_0\pm\sigma_1 A)$ to the both sides of this identity and employing \eqref{eq:sigmaAsquared} we obtain
$$(-8\pm 4i\sigma_1 A)\psi(0_\mp)=0,$$
which is equivalent to
\begin{equation} \label{eq:TCdecoup}
(2i\sigma_1\pm A)\psi(0_\mp)=0.
\end{equation} 
Vice versa, functions obeying \eqref{eq:TCdecoup} clearly satisfy \eqref{eq:TC}. Consequently, $D_A=D_A^+\oplus D_A^-$ with $D_A^\pm$ given by \eqref{eq:decoupOp}.
\end{proof}
\end{proposition}

Next, we will distinguish the following three cases.
\begin{description}[leftmargin=*]
\item[case (i)] $(2i\sigma_1-A)$ is invertible. Then the transmission condition \eqref{eq:TC} reads $\psi(0_+)=\Lambda\psi(0_-)$ with 
\begin{equation} \label{eq:LambdaDef}
\Lambda:=(2i\sigma_1-A)^{-1}(2i\sigma_1+A).
\end{equation}
\item[case (ii)] $(2i\sigma_1+A)$ is invertible. Then the condition \eqref{eq:TC} is equivalent to $\psi(0_-)=\tilde\Lambda\psi(0_+)$ with 
\begin{equation} \label{eq:tildeLambdaDef}
\tilde\Lambda:=(2i\sigma_1+A)^{-1}(2i\sigma_1-A).
\end{equation} 
\item[case (iii)] none of $(2i\sigma_1-A)$ and $(2i\sigma_1-A)$ is invertible .  Then, as was shown in Proposition \ref{prop:decoup}, the condition \eqref{eq:TC} is equivalent to 
\begin{equation} \label{eq:TCdecoup2}
(2i\sigma_1\mp A)\psi_\pm(0_\pm)=0.
\end{equation}
\end{description}
Note that if both matrices $(2i\sigma_1-A)$ and $(2i\sigma_1-A)$ are invertible then the cases (i) and (ii) overlap and $\tilde\Lambda=\Lambda^{-1}$.

\begin{proposition} \label{prop:adjoint}
We have $(D_A)^*=D_{A^*}$. In terms of the matrices $\Lambda$ and $\tilde\Lambda$, which are defined in \eqref{eq:LambdaDef} and \eqref{eq:tildeLambdaDef}, respectively, we get
\begin{align*}
&\begin{aligned}
\dom((D_A)^*)=&\{\psi\equiv\psi_-\oplus\psi_+\in H^1(\R_-;\C^2)\oplus H^1(\R_+;\C^2)|\,\\
& \text{condition \eqref{eq:TC_adjoint} below holds true}\,\}
\end{aligned} \\
&(D_A)^*\psi=\mathscr{D}\psi_-\oplus \mathscr{D}\psi_+,
\end{align*}
where
\begin{equation} \label{eq:TC_adjoint}
\begin{cases}
\sigma_1\Lambda^*\sigma_1\psi(0_+)=\psi(0_-) & \text{in the case (i)}\\
\psi(0_+)=\sigma_1\tilde{\Lambda}^*\sigma_1\psi(0_-) & \text{in the case 
(ii)}.
\end{cases}
\end{equation}
\end{proposition}

\begin{proof}
Recall that for any matrix $A\in\C^{2,2}$ we have 
$$D_{\min}\subset D_A\subset D_{\max}=(D_{\min})^{*},$$
which yields 
$$\overline {D_{\min}}=(D_{\min})^{**}\subset (D_A)^{*}\subset (D_{\min})^{*}=D_{\max}.$$
In particular, $(D_A)^{*}$ is a certain restriction of $D_{\max}$. By splitting integrals over $\R$ into integrals over half-lines $\R_\pm$ and then employing the integration by parts, we deduce that $\forall\varphi\in\dom(D_A)$ and $\forall\psi\in\dom(D_{\max})$ 
\begin{equation*}
\langle D_A\varphi,\psi\rangle=\langle\varphi,D_{\max}\psi\rangle-i(\langle\varphi(0_+),\sigma_1\psi(0_+)\rangle_{\C^2}-\langle\varphi(0_-), \sigma_1\psi(0_-)\rangle_{\C^2}).
\end{equation*}
Therefore, $\psi\in\dom((D_A)^{*})$ if and only if
\begin{equation} \label{eq:boundaryTerms}
\forall\varphi\in\dom(D_A):\, \langle\varphi(0_+),\sigma_1\psi(0_+)\rangle_{\C^2}-\langle\varphi(0_-), \sigma_1\psi(0_-)\rangle_{\C^2}=0.
\end{equation}
We will inspect this condition for the cases (i)--(iii) separately.

In the case (i), we can rewrite \eqref{eq:boundaryTerms} as
\begin{equation*}
\forall\varphi\in\dom(D_A):\, \langle\varphi(0_-),\Lambda^*\sigma_1\psi(0_+)-\sigma_1\psi(0_-)\rangle_{\C^2}=0.
\end{equation*}
Since the values of $\varphi(0_-)$ are arbitrary, we arrive at the first line of \eqref{eq:TC_adjoint}. Substituting for $\Lambda$ from \eqref{eq:LambdaDef} and using the identity $\sigma_1=\sigma_1^{-1}$, we get
$$\sigma_1(-2i\sigma+A^*)(-2i\sigma_1-A^*)^{-1}\sigma_1\psi(0_+)=4i(2i\sigma_0+\sigma_1 A^*)^{-1}\psi(0_+)-\psi(0_+)=\psi(0_-).$$
Applying $(2i\sigma_0+\sigma_1 A^*)$ to the both side of the last equality, we obtain
$$(2i\sigma_1-A^*)\psi(0_+)=(2i\sigma_1+A^*)\psi(0_-),$$
i.e., \eqref{eq:TC} with $A^*$ substituted for $A$. In the case (ii), one would proceed similarly. 

In the case (iii), the operator $D_A$ decouples and \eqref{eq:boundaryTerms} is equivalent to
\begin{equation} \label{eq:boundaryTermsiii}
\forall\varphi_\pm\in\dom(D_A^\pm):\, \langle\varphi_\pm(0), \sigma_1\psi_\pm(0)\rangle_{\C^2}=0,
\end{equation}
where we decomposed $\psi=\psi_-\oplus\psi_+$. Now, $\varphi_\pm (0)$ obey $(2i\sigma_1\mp A)\varphi_\pm(0)=0$. Hence, \eqref{eq:boundaryTermsiii} holds if and only if $\sigma_1\psi_\pm(0)\equiv\sigma_1\psi(0_\pm)\in \ker(2i\sigma_1\mp A)^\bot$. This is equivalent to $\psi(0_\pm)\in\ran(\sigma_1 A^*\pm2i\sigma_0)$. In view of \eqref{eq:TCdecoup2}, it remains to to show that 
\begin{equation*}
\ran(\sigma_1 A^*\pm2i\sigma_0)=\ker(2i\sigma_1\mp A^*).
\end{equation*}
If $\xi\in\ran(\sigma_1 A^*\pm2i\sigma_0)$ then there exists $\eta\in\C^2$ such that $(\sigma_1 A^*\pm2i\sigma_0)\eta=\xi$. Applying $(2i\sigma_1\mp A^*)$ to the both sides of this equation, we get 
$$\mp(4\sigma_1+A^*\sigma_1 A^*)\eta=(2i\sigma_1\mp A^*)\xi.$$
Taking the adjoint of \eqref{eq:sigmaAsquared} we see that the left-hand side is zero, i.e., $\xi\in\ker(2i\sigma_1\mp A^*)$. Vice versa, if we start with $\xi\in\ker(2i\sigma_1\mp A^*)$ then $\pm 2i\xi=\sigma_1 A^*\xi$. Adding $\pm 2i\xi$ to the both sides of this equation we deduce that
$$\xi=\pm\frac{1}{4i}(\sigma_1 A^*\pm 2i\sigma_0)\xi,$$
which means that $\xi\in\ran(\sigma_1 A^*\pm2i\sigma_0)$.
\end{proof}

\begin{corollary} \label{cor:sa}
The operator $D_A$ is self-adjoint if and only if $A$ is hermitian. If $D_A$ is normal then it is necessarily self-adjoint.
\end{corollary}

\begin{proof}
By Proposition \ref{prop:adjoint}, $D_A$ is self-adjoint if and only if $D_A=D_{A^*}$. According to the last statement of Proposition \ref{prop:basic} this happens if and only if $A=A^*$.

The second statement of the theorem follows from the fact that, in view of Proposition \ref{prop:adjoint}, both $D_A$ and $(D_A)^{*}$ are restrictions of $D_{\max}$.
\end{proof}

\subsection{Digression-boundary triplets}  \label{sec:bTriplets}
Every $D_A$ is a closed proper extension of the symmetric operator $D_{\min}$. There is a general framework built on the notion of boundary triplet that may be used to study all such extensions and their properties. For an overview of the subject we refer to \cite{DeMa_85,BrGePa_08,Snoo, Schmudgen}.  Modifying the boundary triple constructed in \cite{PaRi_14}, where $\sigma_2$ instead of $\sigma_1$ appears in the definition of $D_{\min}$, to our setting  we obtain a boundary triplet $(\C^2,\Gamma_1,\Gamma_2)$ for $(D_{\min})^*=D_{\max}$, where
\begin{equation*}
\Gamma_1(\psi):=\sigma_3(\psi(0_-)-\psi(0_+)),\quad \Gamma_2(\psi):=\frac{1}{2}\sigma_2(\psi(0_+)+\psi(0_-)) \quad (\forall\psi\in\dom(D_{\max})).
\end{equation*}
Indeed, it is straightforward to verify that for all $\varphi,\psi\in\dom(D_{\max})$, 
$$\langle\varphi,D_{\max}\psi\rangle-\langle D_{\max}\varphi,\psi\rangle=\langle\Gamma_1\varphi,\Gamma_2\psi\rangle_{\C^2}-\langle\Gamma_2\varphi,\Gamma_1\psi\rangle_{\C^2}$$
by calculations similar to those included in the proof of Proposition \ref{prop:adjoint}, and that the mapping $\psi\in D_{\max}\mapsto(\Gamma_1\psi,\Gamma_2\psi)\in\C^2\times\C^2$ is surjective, i.e., the two defining properties of a boundary triplet hold true.

Now, we rewrite the transmission condition \eqref{eq:TC}  for $\dom(D_A)$ as follows
$$-2i\sigma_1(\psi(0_-)-\psi(0_+))=A(\psi(0_+)+\psi(0_-)).$$
Applying $-\frac{1}{2}\sigma_2$ to the both sides of this equation and using the relation $\sigma_0=\sigma_2^2$, we get the following equivalent equation
\begin{equation} \label{eq:TCtriple}
\Gamma_1\psi=-\sigma_2 A\sigma_2\Gamma_2\psi.
\end{equation}
Therefore, we have 
$$\dom(D_A)=\{\psi\in\dom(D_{\max})|\, (\Gamma_1\psi,\Gamma_2\psi)\in\mathcal{A}\},$$
where $\mathcal{A}:=\{(-\sigma_2 A\sigma_2 y, y)|\, y\in\C^2\}$ is a linear relation on $\C^2$ called the boundary space of $D_A$. By \cite[Lem. 14.6]{Schmudgen}, $(D_A)^*$ is the restriction of $D_{\max}$ whose boundary space  is equal to the adjoint relation $\mathcal{A}^{*}$. This would yield an alternative proof of Proposition \ref{prop:adjoint}. Similarly, the resolvent of $D_A$ that will be obtained in Theorem \ref{theo:res} by guessing followed by a direct verification could have been deduced using the theory of boundary triplets, see Remark \ref{rem:resolvent}. However, to make the text easier to follow also for the readers who are not acquainted with this theory, we preferred to include direct proofs whenever they are not considerably longer or more tedious then application of the boundary triplets techniques (which are apparently indispensable in higher dimensional settings).

\begin{remark} \label{rem:sa}
For every hermitian matrix $A$ there exists a unique unitary matrix $U$ such that $1\notin\sigma(U)$ and 
\begin{equation} \label{eq:Cayley}
A=-i(\sigma_0-U)^{-1}(\sigma_0+U).
\end{equation}
The matrix $U$ is referred to as the Cayley transform of $A$. Note that we added an extra minus sign for future convenience and recall that the Cayley transform is a one-to-one mapping between hermitian matrices and unitary matrices which do not have $\lambda=1$ as an eigenvalue. Plugging \eqref{eq:Cayley} into \eqref{eq:TCtriple} and using the identity $\sigma_0=\sigma_2^2$ we get
\begin{equation*}
\Gamma_1\psi=i(\sigma_0-\tilde U)^{-1}(\sigma_0+\tilde U)\Gamma_2\psi,
\end{equation*}
where $\tilde U:=\sigma_2 U\sigma_2$ is unitary and $1\notin\sigma(\tilde U)$. This is equivalent to 
\begin{equation} \label{eq:TC_U}
\frac{1}{2}(\sigma_0-\tilde U)\Gamma_1\psi=\frac{i}{2}(\sigma_0+\tilde U)\Gamma_2\psi.
\end{equation}
There is one-to-one correspondence between all self-adjoint extensions of $D_{\min}$ and unitary $2\times 2$ matrices that maps $\tilde U\in \C^{2,2}$ to the restriction of $D_{\max}$ to the functions that obey \eqref{eq:TC_U}, cf. \cite[Thm. 1.12]{BrGePa_08} and \cite[Prop. 2]{Pa_06} (the latter result is a special case of \cite[Thm. 1.4 in Chpt. 3]{Gorbachuk}). Since starting with a hermitian $A$, $\tilde U$ is not an arbitrary unitary matrix but a unitary matrix such that $1\notin\sigma(\tilde U)$, not every self-adjoint extension of $D_{\min}$ is described as $D_A$ with a certain (necessarily hermitian) matrix $A$.
\end{remark}

\section{Spectrum and resolvent} \label{sec:spectrum}

\subsection{Free Dirac operator}
Recall that $D_0$ is defined on $\dom(D_0)=H^1(\R;\C^2)$ and, therefore, it is nothing but the one-dimensional free Dirac operator.
It is well known that $D_0$ is self-adjoint and the spectrum of $D_0$ is purely absolutely continuous and equal to
$$\sigma(D_0)=\sigma_{ac}(D_0)=(-\infty,-|m|]\cup[|m|,+\infty).$$
The most direct way how to see it is to employ the Fourier-Plancherel transform, since then $D_0$ is unitarily equivalent to a multiplication operator by a smooth matrix-valued function, cf. \cite[Thm. 1.1]{Thaller} for similar considerations in the three-dimensional setting.

For every $z\notin\sigma(D_0)$, the resolvent $(D_0-z)^{-1}$ is the integral operator with kernel 
\begin{equation}\label{eq:freeResolvent}
R_z(x,y)= \frac{i}{2}(Z(z) + \sgn (x-y)\sigma_1)\ee^{ik(z)|x-y|},
\end{equation}
where 
$$Z(z):=\begin{pmatrix}
\zeta(z) & 0\\
0 & \zeta(z)^{-1}
\end{pmatrix},
\quad \zeta(z):=\frac{z+m}{k(z)},\quad k(z):=\sqrt{z^{2}-m^{2}},$$
cf. \cite{BeDa_94}. Let us stress that for the complex square-root we always adopt the convention that $\Im\sqrt{w}>0$ for all $w\in\C\setminus[0,+\infty)$.

\subsection{Eigenvalues}
Let $z\in\C$. The eigenvalue equation $D_A\psi=z\psi$ yields the following differential equation
\begin{equation*}
\frac{\dd\psi}{\dd x}=i M \psi \quad \text{with} \quad M:=\begin{pmatrix}
0 & z+m\\
z-m & 0
\end{pmatrix},
\end{equation*}
which should hold true at least in the weak sense on $\R_\pm$. The general solution reads as
\begin{equation*}
\begin{split}
\psi(x)&=\exp(ix M)\begin{pmatrix} a\\b \end{pmatrix}\\
&=\begin{cases}
\big(\cos(k(z)x)\sigma_0+ik(z)^{-1}\sin(k(z)x)M\big) \begin{pmatrix} a\\b \end{pmatrix} & \text{if }k(z)\neq 0\\
(\sigma_0+i x M) \begin{pmatrix} a\\b \end{pmatrix} &\text{if }k(z)=0,
\end{cases}
\end{split}
\end{equation*}
where $a,b\in\C$. The case $k(z)=0$ occurs if and only if $z=\pm m$. Then the solution $\psi$ is linear,  and so it is never integrable near infinities, except for the trivial case $a=b=0$. We conclude that $z=\pm m$ are never eigenvalues of $D_A$. 

Let us now focus on the case $k(z)\neq 0$. Writing $k(z)=\eta+i\gamma$,  where $\eta\in\R$ and $\gamma\geq 0$, we get $\psi(x)=(\psi_1(x),\psi_2(x))^T$ with
\begin{equation} \label{eq:EVeqSol}
\begin{split}
&\begin{split}\psi_1(x)=&(a\cosh(\gamma x)-b \zeta(z)\sinh(\gamma x))\cos(\eta x)\\
&+i(b\zeta(z)\cosh(\gamma x)-a\sinh(\gamma x))\sin(\eta x)
\end{split}\\
&\begin{split}\psi_2(x)=&(b\cosh(\gamma x)-a \zeta(z)^{-1}\sinh(\gamma x))\cos(\eta x)\\
&+i(a\zeta(z)^{-1}\cosh(\gamma x)-b\sinh(\gamma x))\sin(\eta x).
\end{split}
\end{split}
\end{equation}
If $\gamma=0$, i.e., when $z^2-m^2>0$, then $\psi$ is not square integrable near infinities, unless $a=b=0$. Hence, there are no eigenvalues in $(-\infty,-|m|)\cup(|m|,+\infty)$. Finally, let $\gamma>0$. Then, $\psi\in L^2(0,+\infty)$ if and only if $a=\zeta(z)b$ and $\psi\in L^2(-\infty,0)$ if and only if $a=-\zeta(z)b$. Substituting these conditions into \eqref{eq:EVeqSol} we get
\begin{equation} \label{eq:IntSol}
\psi(x)=\begin{cases}
a\begin{pmatrix} 1 \\ \zeta(z)^{-1}\end{pmatrix}\ee^{ik(z)x} &\text{for }x>0\\
\tilde a\begin{pmatrix} 1 \\ -\zeta(z)^{-1}\end{pmatrix}\ee^{-ik(z)x} &\text{for }x<0,
\end{cases}
\end{equation}
where $a,\tilde a\in\C$. Now, $z\in\C\setminus((-\infty,-|m|]\cup[|m|,+\infty))$ is an eigenvalue of $D_A$ if and only if there are constants $a,\tilde a\in\C$, at least one of them being non-zero, such that \eqref{eq:IntSol} obeys the transmission condition \eqref{eq:TC}, i.e., 
\begin{equation*}
a(2i\sigma_1-A)\begin{pmatrix} 1 \\ \zeta(z)^{-1}\end{pmatrix}=\tilde a(2i\sigma_1+A)\begin{pmatrix} 1 \\ -\zeta(z)^{-1}\end{pmatrix}.
\end{equation*}
This is equivalent to
\begin{equation} \label{eq:EFcoeff}
(A-2i Z(z)^{-1})\begin{pmatrix} a+\tilde a \\ (a-\tilde a)\zeta(z)^{-1}\end{pmatrix}=0.
\end{equation}
Since $(a+\tilde a, (a-\tilde a)\zeta(z)^{-1})\neq (0,0)$ if and only if $(a,\tilde a)\neq (0,0)$, we infer that our spectral condition is equivalent to $\det(A-2iZ(z)^{-1})=0$. Hence, we arrive at

\begin{proposition} \label{prop:PointSpec}
We have $\sigma_p(D_A)\cap((-\infty,-|m|]\cup[|m|,+\infty))=\emptyset$ and $z\in\C\setminus((-\infty,-|m|]\cup[|m|,+\infty))$ is an eigenvalue of $D_A$ if and only if 
\begin{equation*}
\det\Big(\sigma_0+\frac{i}{2}AZ(z)\Big)=0.
\end{equation*}
In the positive case, the geometric multiplicity of  $z$ equals $\dim (\ker(\sigma_0+\frac{i}{2}AZ(z)))\leq 2$. The associated eigenfunctions are given by \eqref{eq:IntSol}, where $(a,\tilde a)$ are non-trivial solutions of \eqref{eq:EFcoeff}.
\end{proposition}

\subsection{Resolvent}
In the self-adjoint case, i.e., when $A$ is hermitian, the resolvent of $D_A$ is well known and given by the so-called Krein resolvent formula, cf. \cite{BeDa_94} and \cite{AG} for the particular result and a general theory, respectively. For every $z\in\C\setminus\R$, the resolvent $R^A_z=(D_A-z)^{-1}$ is the integral operator with the kernel $R^A_z(x,y)$ of the form
\begin{equation} \label{eq:resAnsatz}
R^A_z(x,y)=R_z(x,y)-R_z(x,0)\mathscr{M}(z)R_z(0,y),
\end{equation}
where $\mathscr{M}(z)$ is a certain $2\times 2$ matrix that depends on $A$ and $z$. We will now use \eqref{eq:resAnsatz} as an ansatz for the resolvent of arbitrary $D_A$  (with $A$ not necessarily hermitian). If $R^A_z(x,y)$ was the integral kernel of $(D_A-z)^{-1}$ then, in particular, 
$$g(x):=\int_\R R^A_z(x,y)\psi(y)\dd y$$
would obey the transmission condition \eqref{eq:TC} for every $\psi\in L^2(\R;\C^2)$. Since 
$$g(0_+)=f(0)-R_z(0_+,0)\mathscr{M}(z)f(0) \quad \text{and} \quad g(0_-)=f(0)-R_z(0_-,0)\mathscr{M}(z)f(0),$$
where we introduced $f:=R_z\psi\in\dom(D_0)=H^1(\R;\C^2)$, this would mean that
\begin{equation*}
2Af(0)=\Big(2i\sigma_1(R_z(0_-,0)-R_z(0_+,0))+A(R_z(0_-,0)+R_z(0_+,0))\Big)\mathscr{M}(z)f(0).
\end{equation*}
Noting that 
\begin{equation*}
R_z(0_-,0)-R_z(0_+,0)=-i\sigma_1 \quad \text{and} \quad R_z(0_-,0)+R_z(0_+,0)=iZ(z),
\end{equation*}
this can be further rearranged as
\begin{equation*}
Af(0)=\Big(\sigma_0+\frac{i}{2}AZ(z)\Big)\mathscr{M}(z)f(0).
\end{equation*}
Assuming that $\sigma_0+\frac{i}{2}AZ(z)$ is invertible and taking the fact that $f(0)\in\C^2$ is arbitrary into account, this would yield
\begin{equation*}
\mathscr{M}(z)=\Big(\sigma_0+\frac{i}{2}AZ(z)\Big)^{-1}A.
\end{equation*} 
These considerations suggest the following result.

\begin{theorem} \label{theo:res}
For every $z\in\C\setminus((-\infty,-|m|]\cup[|m|,+\infty))$ such that $\det(\sigma_0+\frac{i}{2}AZ(z))\neq 0$, the resolvent of $D_A$ is the integral operator $R_z^A$ with kernel
\begin{equation} \label{eq:resolvent}
R_z^A(x,y)=R_z(x,y)-R_z(x,0)\Big(\sigma_0+\frac{i}{2}AZ(z)\Big)^{-1}A R_z(0,y).
\end{equation}
\end{theorem}

\begin{proof}
For the considered values of $z$, $(D_A-z)$ is injective by Proposition \ref{prop:PointSpec}. Moreover, by construction, $R_z^A$ maps $L^2(\R;\C^2)$ to $\dom(D^A)$ and for every $\psi\in L^2(\R;\C^2)$,
$$(D_A-z)R_z^A\psi=\psi,$$
because $((\mathscr{D}-z)R_z(\cdot,0))(x)=0$ for all $x\neq 0$ and $((\mathscr{D}-z)(R_z\psi))(x)=\psi(x)$ for a.e. $x\in\R$ ($R_z(x,0)$ is the fundamental solution of $(\mathscr{D}-z)$). Consequently, $z\in\res(D_A)$ and $(D_A-z)^{-1}=R_z^A$.
\end{proof}

\begin{remark} \label{rem:resolvent}
Alternatively, one can deduce the resolvent $R_z^A$ using the boundary triplet $(\C^2,\Gamma_1,\Gamma_2)$ introduced in Section \ref{sec:bTriplets} and the general Krein-Naimark resolvent formula, cf. \cite[Thm. 14.18]{Schmudgen} or \cite[Cor. 2.6.3]{Snoo}. To this purpose, note that $D_0=D_{\max}\upharpoonright\ker(\Gamma_1)$. For the sake of completeness, we just include the formulae for the gamma field $\gamma(z):=(\Gamma_1\upharpoonright\ker(D_{\max}-z))^{-1}:\, \C^2\to\ker(D_{\max}-z)$ and the Weyl function  $M(z):=\Gamma_2\circ\gamma(z):\,\C^2\to\C^2$ of $D_0$ associated with the triplet $(\C^2,\Gamma_1,\Gamma_2)$  which are building blocks of the Krein-Naimark formula,
\begin{equation*}
(\gamma(z)v)(x)=i R_z(x,0)\sigma_2 v,\quad M(z)v=\frac{i}{2}\sigma_2 Z(z)\sigma_2 v \quad (\forall v\in\C^2,\, \text{a.e. } x\in\R).
\end{equation*}
\end{remark}

\subsection{Full spectrum and spectral transitions}
The point spectrum of $D_A$ is fully described in Proposition \ref{prop:PointSpec}. The condition for $z\in\C\setminus((-\infty,-|m|]\cup[|m|,+\infty))$ to be in $\sigma_p(D_A)$ can be rewritten as follows
\begin{equation} \label{eq:EVeq}
4-\det(A)+2i\tr(AZ(z))=0.
\end{equation}
Since $Z(z)=\mathrm{diag}(\zeta(z),\zeta(z)^{-1})$, it will be useful to know how does the inverse of $\zeta$ look like.

\begin{lemma} \label{lemma:zeta_inv}
Let $m\in\R\setminus\{0\}$ and $\eta\in\C$. Then the equation
\begin{equation} \label{eq:zeta_inv}
\zeta(z)=\eta
\end{equation}
has a solution $z\in\C\setminus((-\infty,-|m|]\cup[|m|,+\infty))$ if and only if
\begin{equation} \label{eq:zeta_inv_cond}
\eta\in\C\setminus\R\quad \wedge \quad  \sgn{m}=\sgn\left(\Im\frac{\eta}{\eta^2-1}\right).
\end{equation}
In the positive case, the solution is unique and given by
\begin{equation} \label{eq:z_sol}
z_\eta=m\frac{\eta^2+1}{\eta^2-1}.
\end{equation}
\end{lemma}

\begin{proof}
By a direct inspection, one verifies that there are no solutions $z$ of \eqref{eq:zeta_inv} for $\eta=\pm 1$ and that the only solution of \eqref{eq:zeta_inv} for $\eta=0$ is $z=-m\in (-\infty,-|m|]\cup[|m|,+\infty)$. Let $\eta\neq\pm 1$ and $\eta\neq 0$. Then squaring \eqref{eq:zeta_inv} we get a quadratic equation in $z$ whose solutions are
\begin{equation*}
z_\pm=\frac{m\pm|m|\sqrt{\eta^4}}{\eta^2-1}=
\begin{cases}
\frac{m\pm|m|\sgn(\Im(\eta^2))\eta^2}{\eta^2-1} & \text{if }\Im(\eta^2)\neq 0\\
\frac{m\pm|m|\sgn(\eta^2)\eta^2}{\eta^2-1} & \text{if }\Im(\eta^2)= 0.
\end{cases}
\end{equation*}
One of these solutions is always $z=-m$ and the other is always given by \eqref{eq:z_sol}. Only the latter may belong to $\C\setminus((-\infty,-|m|]\cup[|m|,+\infty))$. This happens if and only if $\Im(\eta^2)\neq 0$ or $\eta^2<0$, i.e., for $\eta\in\C\setminus\R$. Finally, we check whether $z_\eta$ satisfies the original equation \eqref{eq:zeta_inv}. We have
$$\zeta(z_\eta)=\frac{\eta}{\sgn\left(\Im\frac{m\eta}{\eta^2-1}\right)}.$$
Therefore, we need the second condition in \eqref{eq:zeta_inv_cond} to get $\zeta(z_\eta)=\eta$.
\end{proof}

Now, we will distinguish several cases.
\begin{enumerate}[label=\arabic*),leftmargin=*]
\item $m=0$. For $m=0$, $\zeta(z)=\sgn(\Im z)$ and \eqref{eq:EVeq}  is equivalent to
\begin{equation*}
(\det(A)-4)\sgn(\Im z)=2i\tr(A).
\end{equation*}
Therefore, if, in addition to $m=0$,
\begin{enumerate}[label=1.\alph*)]
\item $\tr(A)=0 \wedge \det(A)=4$ then $\sigma_p(D_A)=\C\setminus\R$.
\item $\tr(A)=0 \wedge \det(A)\neq 4$ then $\sigma_p(D_A)=\emptyset$.
\item $\tr(A)\neq 0 \wedge \det(A)-4=\pm 2i\tr(A)$ then $\sigma_p(D_A)=\C_\pm:=\{z\in\C|\, \pm\Im z>0\}$.
\item $\tr(A)\neq 0 \wedge \det(A)-4\neq\pm 2i\tr(A)$ then $\sigma_p(D_A)=\emptyset$.
\end{enumerate}
\item $m\neq 0$. For $m\neq 0$, it is convenient to rewrite \eqref{eq:EVeq} as follows
\begin{equation} \label{eq:quadratic}
\alpha\zeta(z)^2+i\frac{\det(A)-4}{2}\zeta(z)+\delta=0,
\end{equation}
where $(\alpha,\delta)$ is the diagonal of $A$. Hence, if, in addition to $m\neq 0$,
\begin{enumerate}[label=2.\alph*)]
\item $\alpha=\delta=0 \wedge \det(A)=4$ then  $\sigma_p(D_A)=\C\setminus((-\infty,-|m|]\cup[|m|,+\infty))$.
\item $\alpha=\delta=0 \wedge \det(A)\neq 4$ then $\sigma_p(D_A)=\emptyset$.
\item $\alpha=0 \wedge \delta\neq 0 \wedge \det(A)=4$ then $\sigma_p(D_A)=\emptyset$.
\item $\alpha=0 \wedge \delta\neq 0 \wedge \det(A)\neq 4$ then we get
\begin{equation*} 
\zeta(z)=i\frac{2\delta}{\det(A)-4}.
\end{equation*}
Using Lemma \ref{lemma:zeta_inv}, we see that there is at most one eigenvalue of $D_A$.
\item $\alpha\neq 0$ then \eqref{eq:quadratic} is a quadratic equation for $\zeta(z)$. In view of Lemma \ref{lemma:zeta_inv}, we conclude that in this case there are at most two eigenvalues of $D_A$.
\end{enumerate}
\end{enumerate}

\begin{remark}
In the cases 2.d) and 2.e), one can use \eqref{eq:z_sol} to get rather lengthy but still  explicit formulae for the eigenvalues of $D_A$. Consequently, one can calculate weak and strong coupling asymptotic expansions for these eigenvalues. In particular, one sees that the leading term of every eigenvalue of $D_{\kappa A}$ is either $m$ or $-m$ as $\kappa\to 0$.
\end{remark}

According to Theorem \ref{theo:res} and Proposition \ref{prop:PointSpec}, 
\begin{equation} \label{eq:specInclusion}
\sigma(D_A)\subset (-\infty,-|m|]\cup[|m|,+\infty)\cup\sigma_p(D_A),
\end{equation}
where $((-\infty,-|m|]\cup[|m|,+\infty))\cap\sigma_p(D_A)=\emptyset$.
Now, $D_0$ is self-adjoint, $\sigma(D_0)=\sigma_{ess}(D_0)=(-\infty,-|m|]\cup[|m|,+\infty)$, and the second term on the right-hand side of \eqref{eq:resolvent} is a Hilbert Schmidt (and thus compact) operator. If we have either
\begin{enumerate}[label=\alph*)]
\item $\sigma(D_0)\neq \R \wedge \res(D_A)\neq \emptyset$

\noindent or
\item $\res(D_A)\cap\C_\pm\neq\emptyset$
\end{enumerate}
then, by the Weyl essential spectrum theorem \cite[Thm. XIII.14]{RS4}, $\sigma_{ess}(D_A)=\sigma_{ess}(D_0)$. Recall that, in this result, $\sigma_{ess}(D_A):=\sigma(D_A)\setminus\sigma_d(D_A)$, where the discrete spectrum $\sigma_d$ is defined  as the set of all isolated (in the spectrum) eigenvalues with finite algebraic multiplicity, cf. \cite[Sect. XII.2]{RS4}.
From the analysis of $\sigma_p(D_A)$ above, we see that either a) or b) holds true except for the cases 1.a), 1.c), and 2.a). In the case 1.a) and 2.a) we get $\sigma(D_A)=\C$, because the spectrum is always a closed set. By the same argument and taking \eqref{eq:specInclusion} into account we get either $\sigma(D_A)=\overline{\C_+}$ or $\sigma(D_A)=\overline{\C_-}$ in the case 1.c). In all cases 1.a)--2.e), there is equality in \eqref{eq:specInclusion}, i.e., Theorem \ref{theo:res} provides the resolvent formula for all $z\in\res(D_A)$.

We summarize our findings in the following theorem that extends Proposition \ref{prop:PointSpec}.

\begin{theorem} \label{theo:spec}
We have
\begin{equation*}
\sigma(D_A)=(-\infty,-|m|]\cup[|m|,+\infty)\cup\sigma_p(D_A),
\end{equation*}
where no points of $\sigma_p(D_A)$ are contained in $(-\infty,-|m|]\cup[|m|,+\infty)$ and $z\in\C\setminus((-\infty,-|m|]\cup[|m|,+\infty))$ is an eigenvalue of $D_A$ if and only if 
\begin{equation} \label{eq:ev_DA}
\det\Big(\sigma_0+\frac{i}{2}AZ(z)\Big)=0.
\end{equation}
In particular, if
\begin{itemize}[leftmargin=2em]
 \item $m=0,\, \tr(A)=0,$ and $\det(A)=4$ then $\sigma_p(D_A)=\C\setminus\R$.
 \item $m=0,\, \tr(A)\neq 0,$ and  $\det(A)-4=\pm 2i\tr(A)$ then $\sigma_p(D_A)=\C_\pm$.
 \item $m\neq 0$, $A$ has zero diagonal, and $\det(A)=4$ then $\sigma_p(D_A)=\C\setminus((-\infty,-|m|]\cup[|m|,+\infty))$.
\end{itemize}
In all other cases there are at most two eigenvalues of $D_A$. Especially, under conditions 1.b), 1.d), 2.b), and 2.c) (with appropriate restrictions on values of $m$), $D_A$ has no eigenvalues. In all cases, the geometric multiplicity of the eigenvalue $z$ equals $\dim (\ker(\sigma_0+\frac{i}{2}AZ(z)))\leq 2$, the algebraic multiplicity of an isolated eigenvalue is finite, and, for all $z\notin\sigma(D_A)$, the resolvent $(D_A-z)^{-1}$ is given by \eqref{eq:resolvent}.
\end{theorem}

\begin{corollary} \label{cor:spec}
Let $m\neq 0$ and $z$ be an eigenvalue of $D_A$. Then $z$ is simple except for the choice 
\begin{equation} \label{eq:A_deg}
A=\begin{pmatrix}
\alpha & 0\\
0 & -\frac{4}{\alpha}
\end{pmatrix},\quad \alpha\in\C\setminus\{0\},
\end{equation} 
(that falls under the decoupled case (iii) described below the proof of Proposition \ref{prop:decoup}) in which case the geometric multiplicity of $z$ equals $2$ and there are no other eigenvalues of $D_A$.  
\end{corollary}

\begin{proof}
The geometric multiplicity of $z$ in the spectrum of $D_A$ equals two if and only if $\dim (\ker(\sigma_0+\frac{i}{2}AZ(z)))=2$, i.e.,  $AZ(z)=2i\sigma_0$. This is equivalent to
\begin{equation} \label{eq:A_cond} 
\beta=\gamma=0 \quad\wedge\quad \zeta(z)\alpha=\zeta(z)^{-1}\delta=2i.
\end{equation}
Consequently, $\alpha\delta=\det(A)=-4$, $A$ is of the form \eqref{eq:A_deg}, and we get \eqref{eq:decoup_cond}, i.e., we are necessarily in the decoupled case (iii). 

Vice versa,  let $A$ be as in \eqref{eq:A_deg} and $z$ be arbitrary for now. Then $\zeta(z)=2i/\alpha$  is a double root of \eqref{eq:quadratic} and \eqref{eq:A_cond} is satisfied. Putting this together with  Lemma \ref{lemma:zeta_inv}, we conclude that there is at most one eigenvalue of $D_A$ and its geometric multiplicity equals $2$.
\end{proof}

We see that the spectrum of $D_A$ may change abruptly with an infinitesimal change of the matrix $A$. For example, put $m=0$ and
$$A=\begin{pmatrix}
 i\kappa & 2+\varepsilon\\
 -2 & 0
\end{pmatrix},
$$
where $\kappa,\,\varepsilon\in\C$. Then we get
$$\sigma_p(D_A)=
\begin{cases}
\C\setminus\R & \text{if }\kappa=\varepsilon=0,\\
\emptyset & \text{if }\kappa=0 \wedge \varepsilon\neq 0\text{ or }0\neq \kappa\neq\pm\varepsilon,\\
\C_\pm & \text{if }\varepsilon=\mp\kappa\neq 0.\\
\end{cases}
$$
Another interesting spectral transition occurs when $m\neq 0$ and 
$$A=\begin{pmatrix}
 0 & 4\kappa\\
 -\kappa^{-1} & \delta
\end{pmatrix}
$$
with $\kappa\in\C\setminus 0$ and $\delta\in\C$, because then
$$\sigma_p(D_A)=
\begin{cases}
\C\setminus((-\infty,-|m|]\cup[|m|,+\infty)) & \text{if }\delta=0,\\
\emptyset & \text{if }\delta\neq 0.
\end{cases}
$$

\section{Approximations by regular non-local potentials} \label{sec:approx}
We start this section with introducing a useful convention. Let $\langle\cdot,\cdot\rangle$ stands for the inner product on $L^2(\R;\C)$. Then with some abuse of notation we define
$$\langle f|\psi\rangle = \begin{pmatrix}
\langle f|\psi_1\rangle\\
\langle f|\psi_2\rangle
\end{pmatrix}, \quad
\langle f|B\rangle =\begin{pmatrix}
\langle f|B_{11}\rangle & \langle f|B_{12}\rangle\\
\langle f|B_{21}\rangle & \langle f|B_{22}\rangle
\end{pmatrix},$$
for all
$$f\in L^{2}(\R;\C),\quad \psi\equiv\begin{pmatrix} \psi_1 \\ \psi_2\end{pmatrix}\in L^{2}(\R;\C^2),\quad B\equiv\begin{pmatrix}B_{11} & B_{12}\\ B_{21} & B_{22}\end{pmatrix}\in L^{2}(\R;\C^{2,2}).$$

Now, we will construct a family of approximating operators for $D_A$. For $\varepsilon>0$ and $v\in L^1(\R;\R)\cap L^2(\R;\R)$ such that $\int_\R v=1$, we  put $v_\varepsilon(x):=\varepsilon^{-1}v(\varepsilon^{-1}x)$ and
$$W_\varepsilon:=|v_\varepsilon\rangle\langle v_\varepsilon|,$$
i.e., $W_\varepsilon \psi=\langle v_\varepsilon,\psi\rangle v_\varepsilon$, where $\psi$ may be scalar, vector, or even matrix valued $L^2$-function. Note that 
\begin{equation*}
\int_\R v_\varepsilon(x)\dd x=1
\end{equation*}
and $v_\varepsilon$ converges to $\delta_0$  in the sense of distributions as $\varepsilon\to 0$.  Therefore, we have
\begin{equation} \label{eq:formal_lim}
\lim_{\varepsilon\to 0}\langle W_\varepsilon f, g\rangle=\lim_{\varepsilon\to 0}\langle f,v_\varepsilon\rangle \langle v_\varepsilon,g\rangle=f(0) g(0)=((\delta_0,f)\delta_0,g),
\end{equation}
for all $f,g\in C_0^\infty(\R;\R)$. Since the right-hand side of \eqref{eq:formal_lim} may be formally written as $\langle|\delta_0\rangle\langle\delta_0|f,g\rangle$, this suggests that
\begin{equation*}
\begin{split}
&\dom(D_A^\varepsilon):=\dom(D_0)=H^1(\R;\C^2)\\
&D_A^\varepsilon=D_0+A\otimes W_\varepsilon
\end{split}
\end{equation*}
may be a good candidate for the sought approximation.

\subsection{Norm resolvent convergence}

\begin{theorem} \label{theo:approx}
Let $z\in\res(D_A)$. Then for all $\varepsilon>0$ small enough, $z\in\res(D_A^\varepsilon)$ and 
\begin{equation*}
\lim_{\varepsilon\to 0}\|(D_A^\varepsilon-z)^{-1}-(D_A-z)^{-1}\|=0.
\end{equation*}
\end{theorem}

\begin{remark}
For $A$ and $m$ such that $\det(A)=4$ and either $m=0$ and $\tr(A)=0$ or $m\neq 0$ and $A$ has zero diagonal, $\res(D_A)=\emptyset$, see Theorem \ref{theo:spec}. Therefore, in these exceptional cases it does not make sense to speak about approximations of the resolvent.
\end{remark}

The rest of the section is devoted to the proof of Theorem \ref{theo:approx}. We start with several auxiliary results.

\begin{lemma} \label{lemm:lim_inv_matrix}
Let $z\in\res(D_A)$. Then for all $\varepsilon\in(0,\varepsilon_z)$, where $\varepsilon_z>0$ is a $z$-dependent constant, the matrix $(\sigma_0+\langle v_\varepsilon, A R_z v_\varepsilon \rangle)$ is invertible and
\begin{equation*} 
\lim_{\varepsilon\to 0}(\sigma_0+\langle v_\varepsilon, A R_z v_\varepsilon \rangle)^{-1}=\Big(\sigma_0+\frac{i}{2}AZ(z)\Big)^{-1}.
\end{equation*}
\end{lemma}

\begin{proof}
By Theorem \ref{theo:spec}, $(\sigma_0+\frac{i}{2}AZ(z))$ is invertible for every  $z\in\res(D_A)$. Once we show that 
\begin{equation} \label{eq:mat_lim}
\lim_{\varepsilon\to 0}\langle v_\varepsilon, A R_z v_\varepsilon \rangle=\frac{i}{2}AZ(z),
\end{equation}
the lemma follows.
Using the explicit formula for $R_z$ we obtain
\begin{equation} \label{eq:product}
\langle v_\varepsilon, A R_z v_\varepsilon \rangle=A\langle v_\varepsilon,  R_z v_\varepsilon \rangle=\frac{i}{2}A 
\begin{pmatrix}
\zeta(z) \alpha_\varepsilon & \beta_\varepsilon\\
\beta_\varepsilon & \zeta(z)^{-1}\alpha_\varepsilon
\end{pmatrix},
\end{equation}
where
\begin{equation*}
\alpha_{\varepsilon}:=\int_{\R^{2}}v_{\varepsilon}(x)\ee^{ik(z)|x-y|}v_{\varepsilon}(y)\dd x\dd y,\,\,
\beta_{\varepsilon}:=\int_{\R^{2}}\sgn(x-y)v_{\varepsilon}(x)\ee^{ik(z)|x-y|}v_{\varepsilon}(y)\dd x\dd y.
\end{equation*}
We observe that $\beta_\varepsilon=0$ because of the antisymmetry of integrand and that
\begin{equation} \label{eq:alpha_lim}
\lim_{\varepsilon\to 0}\alpha_\varepsilon=\lim_{\varepsilon\to 0}\int_{\R^2} v(x)\ee^{i \varepsilon k(z)|x-y|}v(y)\dd x \dd y=1,
\end{equation}
due to the dominated convergence theorem and the fact that $\int_\R v=1$. Hence, we get \eqref{eq:mat_lim}.
\end{proof}

\begin{proposition} \label{prop:approx_res}
Let $z\in\C\setminus((-\infty,-|m|]\cup[|m|,+\infty))$ be such that $\det(\sigma_0+\langle v_{\varepsilon},A R_{z}v_{\varepsilon}\rangle)\neq 0$. Then  $z\in\res(D_A^\varepsilon)$ and $(D_A^\varepsilon-z)^{-1}$ is the integral operator with the kernel
\begin{equation} \label{eq:res_kernel_approx}
R_z^{A,\varepsilon}(x,y):=R_{z}(x,y)-\int_{\mathbb{R}^{2}}R_z (x,\varepsilon s)v(s)(\sigma_0+\langle v_{\varepsilon},A R_{z}v_{\varepsilon}\rangle)^{-1} A v(t)R_z(\varepsilon t, y)\dd s \dd t.
\end{equation}
\end{proposition}

\begin{proof}
For $z\in\C\setminus((-\infty,-|m|]\cup[|m|,+\infty))=\res(D_0)$, we have
$$(D_A^\varepsilon-z)^{-1}=R_{z}(I+A\otimes W_\varepsilon R_{z})^{-1},$$
whenever the inverse on the right-hand side exists. To find this inverse, take any $g\in L^2(\R;\C^2)$ and consider the equation
\begin{equation} \label{eq:invertingRF}
(I+A\otimes W_\varepsilon R_{z})\psi=\psi+\langle v_\varepsilon, A R_z\psi\rangle v_\varepsilon=g
\end{equation}
for $\psi\in L^2(\R;\C^2)$. Applying $A R_z$ to the both sides of \eqref{eq:invertingRF} and taking  the inner product of the result with $v_\varepsilon$ we get
$$(\sigma_0+\langle v_\varepsilon, A R_z v_\varepsilon \rangle)\langle v_\varepsilon, A R_z\psi\rangle=\langle v_\varepsilon, A R_z g\rangle.$$
Since the matrix $(\sigma_0+\langle v_\varepsilon, A R_z v_\varepsilon \rangle)$ is assumed to be invertible, we can calculate $\langle v_\varepsilon, A R_z\psi\rangle$ from this equation and substitute for it in \eqref{eq:invertingRF}. This yields
$$\psi=g-(\sigma_0+\langle v_\varepsilon, A R_z v_\varepsilon \rangle)^{-1} \langle v_\varepsilon, A R_z g \rangle v_\varepsilon$$
and, consequently,
\begin{equation} \label{eq:res_approx}
(D_A^\varepsilon-z)^{-1}=R_z-R_z (\sigma_0+\langle v_\varepsilon, A R_z v_\varepsilon \rangle)^{-1}A\otimes W_\varepsilon R_z.
\end{equation}
Using an obvious substitution in the expression for the integral kernel of the right-hand side of \eqref{eq:res_approx} we arrive at \eqref{eq:res_kernel_approx}.
\end{proof}

Below, we will denote  the Frobenius norm of the matrix $B\equiv (B_{ij})_{i,j=1}^{n,m}$ by $|B|$, 
$$|B|^{2}=\sum_{i,j=1}^{n,m}|B_{ij}|^{2}.$$
Recall that the Frobenius norm is sub-multiplicative, i.e., for all $B\in \C^{n,m}$ and $C\in \C^{m,k}$, $|BC|\leq|B| |C|$.

\begin{lemma}\label{lemm:supR_z}
For every $z\in\res(D_0)$, $\sup_{y\in\mathbb{R}}\int_\mathbb{R}|R_z(x,y)|^{2}\dd x < +\infty$.
\end{lemma}

\begin{proof}
First, note that $\forall z\in\res(D_0)$, $\Im k(z)>0$.
Using the triangle inequality, we get
\begin{multline} \label{eq:res_ker_est}
|R_z(x,y)| =\Big|\frac{i}{2}(Z(z)+\sgn (x-y)\sigma_1 )\ee^{ik(z)|x-y|}\Big|\\
\leq \frac{1}{2}(|Z(z)|+|\sigma_1|)\big|\ee^{ik(z)|x-y|}\big|
=\frac{1}{2}(|Z(z)|+|\sigma_1|)\ee^{-\Im k(z)|x-y|}.
\end{multline}
This implies that
\begin{equation*}
\int_\mathbb{R} |R_z(x,y)|^{2}\dd x \leq C_z \int_\mathbb{R}\ee^{-2\Im k(z)|x-y|}\dd x = C_z \int_\mathbb{R}\ee^{-2\Im k(z)|x|}\dd x,
\end{equation*}
where $C_z>0$ is a $z$-dependent constant.  The final bound is finite and clearly $y$-independent. 
\end{proof}

\begin{lemma} \label{lemm:res_diff}
For every $z\in\res(D_0)$,
$$\lim_{\varepsilon\to 0}\int_\R\left(\int_\R |R_z(x,\varepsilon s) - R_z(x,0)||v(s)|\dd s\right)^{2}\dd x = 0.$$
\end{lemma}
\begin{proof}
Put $r_\varepsilon(s):=\int_\R|R_z(x,\varepsilon s) - R_z(x,0)|^{2}\dd x$.  Using the Minkowski integral inequality we deduce that
\begin{equation} \label{eq:Minkowski}
\int_\R\left(\int_\R |R_z(x,\varepsilon s) - R_z(x,0)||v(s)|\dd s\right)^{2}\dd x \leq  \left(\int_\R r_\varepsilon(s)^{\frac{1}{2}} |v(s)|\dd s \right)^{2}.
\end{equation}
Since, by the triangle and the Young inequality, 
$$r_\varepsilon(s)\leq \int_\R (|R_z(x,\varepsilon s)|+|R_z(x,0)|)^{2}\dd x\leq 2\int_\R |R_z(x,\varepsilon s)|^2+|R_z(x,0)|^2 \dd x,$$
we infer from Lemma \ref{lemm:supR_z} that $r_\varepsilon$ is bounded uniformly in $\varepsilon\in(0,+\infty)$. 
Consequently, using the dominated convergence theorem, we get
\begin{equation} \label{eq:firstLeb}
\lim_{\varepsilon\to 0}\int_\R r_\varepsilon(s)^{\frac{1}{2}} |v(s)|\dd s=\int_\R \lim_{\varepsilon\to 0}r_\varepsilon(s)^{\frac{1}{2}} |v(s)|\dd s. 
\end{equation}

Let us now look at the limit on the right-hand side of \eqref{eq:firstLeb}. For any $s\in\R$ and a fixed constant $\delta>0$ we will consider only  $\varepsilon\in (0,\delta/|s|)$ (for $s=0$, the right endpoint should be understood as $+\infty$). 
Similarly as in \eqref{eq:res_ker_est}, we deduce that
\begin{equation*} 
|R_z(x,\varepsilon s)|^{2}\leq C_z \ee^{-2\Im k(z)|x-\varepsilon s|}\leq C_z \ee^{-2\Im k(z)(|x|-\delta)}
\end{equation*}
with a $z$-dependent constant $C_z>0$. Hence, we have
\begin{multline*}
|R_z(x,\varepsilon s) - R_z(x,0)|^{2}\\
\leq 2(|R_z(x,\varepsilon s)|^2+|R_z(x,0)|^2 )\leq 2(C_z \ee^{-2\Im k(z)(|x|-\delta)}+|R_z(x,0)|^2 ).
\end{multline*}
This upper bound is  $\varepsilon$-independent and belongs to $L^1(\R;\dd x)$.  By the dominated convergence theorem, we arrive at
\begin{equation*}
\lim_{\varepsilon\to 0}r_\varepsilon(s)=\int_\R \lim_{\varepsilon\to 0}|R_z(x,\epsilon s) - R_z(x,0)|^{2}\dd x=0.
\end{equation*}
Putting this together with \eqref{eq:Minkowski} and \eqref{eq:firstLeb} we obtain the assertion of the lemma.
\end{proof}

\noindent\textit{Proof of Theorem \ref{theo:approx}.}
The difference $K_z^\varepsilon:=((D_A-z)^{-1}-(D_A^\varepsilon-z)^{-1})$ is the integral operator with the kernel 
\begin{equation*}
K_z^\varepsilon(x,y):=\int_{\R^{2}}R_z (x,\varepsilon s)v(s)T_\varepsilon(z) A v(t)R_z(\varepsilon t, y)\dd s \dd t-R_z(x,0)T(z)A R_z(0,y),
\end{equation*}
where 
\begin{equation*}
T_\varepsilon(z):=(\sigma_0+\langle v_{\varepsilon},A R_{z}v_{\varepsilon}\rangle)^{-1}, \quad  T(z):=\Big(\sigma_0+\frac{i}{2}AZ(z)\Big)^{-1}.
\end{equation*}
We will show that the Hilbert-Schmidt norm of $K_z^\varepsilon$, that cannot be smaller than the operator norm of $K_z^\varepsilon$, tends to zero as $\varepsilon\to 0$. 

Since $\int_R v=1$, we may write
\begin{multline*}
\int_{\R^2}|K_z^\varepsilon(x,y)|^2\dd x\dd y
=\int_{\R^2}\Big|\int_{\R^{2}}\big(R_z (x,\varepsilon s)v(s)T_\varepsilon(z) A v(t)R_z(\varepsilon t, y)\\
-R_z(x,0)v(s)T(z)Av(t) R_z(0,y)\big)\dd s\dd t\Big|^2\dd x\dd y.
\end{multline*}
Using the triangle inequality followed by the inequality $(a+b+c)^2\leq 3(a^2+b^2+c^2)\, (\forall a,b,c\in\R)$ we get
\begin{multline*}
\int_{\R^2}|K_z^\varepsilon(x,y)|^2\dd x\dd y\leq \int_{\R^2}\Big(\int_{\R^{2}}\big|R_z (x,\varepsilon s)v(s)T_\varepsilon(z) A v(t)R_z(\varepsilon t, y)\\
-R_z(x,0)v(s)T(z)Av(t) R_z(0,y)\big|\dd s\dd t\Big)^2\dd x\dd y\leq 3(I_1+I_2+I_3),
\end{multline*}
where
\begin{align*}
&I_1:= \int_{\R^2}\Big(\int_{\R^{2}}\big|(R_z (x,\varepsilon s)-R_z(x,0))v(s)T_\varepsilon(z) A v(t)R_z(\varepsilon t, y)\big|\dd s\dd t\Big)^2\dd x\dd y,\\
&I_2:=\int_{\R^2}\Big(\int_{\R^{2}}\big|R_z(x,0)v(s)(T_\varepsilon(z)-T(z))A v(t)R_z(\varepsilon t, y)\big|\dd s\dd t\Big)^2\dd x\dd y,\\
&I_3:=\int_{\R^2}\Big(\int_{\R^{2}}\big|R_z(x,0)v(s)T(z) A v(t)(R_z(\varepsilon t, y)-R_z(0,y))\big|\dd s\dd t\Big)^2\dd x\dd y.\\
\end{align*}
We will prove that every $I_j$ tends to zero as $\varepsilon\to 0$.

By the Fubini theorem and sub-multiplicativity of the Frobenius norm,
\begin{multline*}
I_1\leq\int_{\R}\Big(\int_{\R}|(R_z (x,\varepsilon s)-R_z(x,0))| |v(s)|\dd s\Big)^2\dd x\\
\times \int_\R\Big(\int_\R |T_\varepsilon(z) A| |v(t)| |R_z(\varepsilon t, y)|\dd t\Big)^2\dd y.
\end{multline*}
Since, according to Lemma \ref{lemm:lim_inv_matrix}, $\lim_{\varepsilon\to 0}T_\varepsilon(z)=T(z)$, the term $|T_\varepsilon(z)A|$ is bounded uniformly in $\varepsilon$ on a right neighbourhood of zero. Using the Minkowski integral inequality and then Lemma \ref{lemm:supR_z} (it holds true also with the roles of $x$ and $y$ interchanged), we obtain
\begin{equation} \label{eq:mixed_est}
\int_\R\Big(\int_\R |v(t)| |R_z(\varepsilon t, y)|\dd t\Big)^2\dd y\leq \Big(\int_\R |v(t)| \Big(\int_\R |R_z(\varepsilon t, y)|^2\dd y\Big)^{1/2}\dd t\Big)^2\leq C_z
\end{equation}
with a constant $C_z$ which depends on $z$ but is $\varepsilon$-independent.
Putting this together with Lemma \ref{lemm:res_diff}, we conclude that $\lim_{\varepsilon\to 0} I_1=0$. Similarly, one shows that $\lim_{\varepsilon\to 0} I_3=0$.

Finally, by the Fubini theorem and sub-multiplicativity of the Frobenius norm,
\begin{multline} \label{eq:I3}
I_2\leq|(T_\varepsilon(z)-T(z))A|^2 \Big(\int_\R|v(s)|\dd s\Big)^2\int_{\R}|R_z(x,0)|^2\dd x\\
\times \int_\R\Big(\int_\R |v(t)| |R_z(\varepsilon t, y)|\dd t\Big)^2\dd y.
\end{multline}
Due to Lemma \ref{lemm:lim_inv_matrix}, $\lim_{\varepsilon\to 0}|(T_\varepsilon(z)-T(z))A|=0$. The first and the second integral in \eqref{eq:I3} are bounded (and clearly $\varepsilon$-independent), since $v\in L^1(\R)$ and Lemma \ref{lemm:supR_z} holds true, respectively. The last integral in \eqref{eq:I3} is estimated by an $\varepsilon$-independent constant in \eqref{eq:mixed_est}. Therefore, $\lim_{\varepsilon\to 0} I_2=0$, too.
\hfill $\square$

\vspace{0.5em}

\subsection{Spectrum of approximations}
The eigenvalue equation for $D_A^\varepsilon$ can be written as
\begin{equation*}
(D_0-z)\psi=-A\langle v_\varepsilon,\psi\rangle v_\varepsilon
\end{equation*}
Let $z\in\C\setminus((-\infty,-|m|]\cup[|m|,+\infty))$. Then this is equivalent to
\begin{equation*}
\psi=-R_z A\langle v_\varepsilon,\psi\rangle v_\varepsilon.
\end{equation*}
Therefore, every eigenfunction is of the form $\psi=R_z A\xi v_\varepsilon$ with some $\xi\in\C^2$ such that $A\xi\neq 0$. Consequently, $z$ is an eigenvalue if and only if there exists $\xi\notin\ker(A)$  such that $R_z A\xi v_\varepsilon=-R_z A\langle v_\varepsilon,R_z v_\varepsilon\rangle A\xi v_\varepsilon$, i.e.,
\begin{equation} \label{eq:ev_approx_eq}
(\sigma_0+A\langle v_\varepsilon,R_z v_\varepsilon\rangle) A\xi=0.
\end{equation}
If this holds then clearly there exists $\tau\in\C^2\setminus\{0\}$ such that
\begin{equation} \label{eq:ev_eq_eta}
(\sigma_0+A\langle v_\varepsilon,R_z v_\varepsilon\rangle) \tau=0.
\end{equation}
The corresponding eigenfunctions would be of the form $\psi=R_z\tau v_\varepsilon$. On the other hand, if there is $\tau\in\C^2\setminus\{0\}$ such that \eqref{eq:ev_eq_eta} holds then $\tau=-A\langle v_\varepsilon,R_z v_\varepsilon\rangle\tau$, and so $\xi:=\langle v_\varepsilon,R_z v_\varepsilon\rangle\tau\notin\ker(A)$ and
$$(\sigma_0+A\langle v_\varepsilon,R_z v_\varepsilon\rangle)A\xi=-(\sigma_0+A\langle v_\varepsilon,R_z v_\varepsilon\rangle)\tau=0,$$
i.e., \eqref{eq:ev_approx_eq} holds true.
Recalling formula \eqref{eq:product} (together with  the fact that $\beta_\varepsilon=0$) and noting that
\begin{equation*} 
\alpha_\varepsilon(k)=\alpha_1(\varepsilon k)
\end{equation*}
for 
$$\alpha_{\varepsilon}(k)\equiv\alpha_\varepsilon=\int_{\R^{2}}v_{\varepsilon}(x)\ee^{ik|x-y|}v_{\varepsilon}(y)\dd x\dd y,$$
we get the following result.

\begin{proposition} \label{prop:ev}
Let $\varepsilon>0$. Then $z\in\C\setminus((-\infty,-|m|]\cup[|m|,+\infty))$ is an eigenvalue of $D_A^\varepsilon$ if and only if 
\begin{equation} \label{eq:ev_DAeps}
\det\Big(\sigma_0+\frac{i}{2} \alpha_1(\varepsilon k(z)) A Z(z) \Big)=0.
\end{equation}
The geometric multiplicity of such $z$ equals $\dim\ker\left(\sigma_0+\frac{i}{2} \alpha_1(\varepsilon k(z)) A Z(z) \right)\leq 2$ and the associated eigenfunctions are of the form
$$\psi=R_z\tau v_\varepsilon \text{ with } \tau\in\ker\Big(\sigma_0+\frac{i}{2} \alpha_1(\varepsilon k(z)) A Z(z) \Big)\setminus\{0\}.$$
\end{proposition}

Due to the norm-resolvent convergence result of Theorem \ref{theo:approx}, the spectrum of the limiting operator $D_A$ cannot suddenly expand nor contract as $\varepsilon\to 0$ in the self-adjoint setting, i.e., when $A$ is hermitian, cf. \cite[Thm. VIII.23]{RS1}. In the non-self-adjoint case, only the principle of non-contraction is preserved in general, see \cite[Sect. IV.\S 3, Sect. 1--2]{kato}. (Note that in the present setting, one may get the non-contraction principle directly from Proposition \ref{prop:approx_res}.)

Nevertheless, for a finite system of eigenvalues, the principle of non-expansion also holds true \cite[Chpt. IV, \S 3, Sect. 5]{kato}. In particular, if $z\in\C$ is a discrete eigenvalue of $D_A$ separated by a closed curve $\Gamma$ from the rest of $\sigma(D_A)$ then, for all $\varepsilon$ sufficiently small, the part of $\sigma(D_A^\varepsilon)$ enclosed by $\Gamma$ consists of eigenvalues whose total algebraic multiplicity equals to the algebraic multiplicity of $z$.

Let us now inspect the spectrum of $D_A^\varepsilon$ in more detail. Recall that the eigenvalues of $D_A^\varepsilon$ are given by the zeros of the function
$$\eta_\varepsilon(z):=\det\Big(\sigma_0+\frac{i}{2} \alpha_1(\varepsilon k(z)) A Z(z) \Big),$$
which is defined  on $\C\setminus((-\infty,-|m|]\cup[|m|,+\infty))$.

\begin{lemma} \label{lemma:analyticity}
The function $\eta_\varepsilon$ is analytic on $\C\setminus((-\infty,-|m|]\cup[|m|,+\infty))$.
\end{lemma} 

\begin{proof}
It is sufficient to scrutinize analyticity of $k(z)$, $\zeta(z)$, and $\zeta^{-1}(z)$. Since $w\mapsto\sqrt{w}$ is analytic on $\C\setminus[0,+\infty)$, $k(z)=\sqrt{z^2-m^2}$ is analytic on $\C\setminus((-\infty,-|m|]\cup[|m|,+\infty))$. The function $\zeta(z)$ is given by a ratio of two non-zero analytic functions on $\C\setminus((-\infty,-|m|]\cup[|m|,+\infty))$. Therefore, both $\zeta(z)$ and $\zeta^{-1}(z)$ are analytic on the same set.
\end{proof}

\begin{proposition} \label{prop:app_spec}
Let $\varepsilon>0$. Then the following holds.
\begin{enumerate}[label=\alph*),leftmargin=*]
\item There are at most countably many eigenvalues of $D_A^\varepsilon$ in $\C\setminus((-\infty,-|m|]\cup[|m|,+\infty))$ and they may accumulate only  in $(-\infty,-|m|]\cup[|m|,+\infty)$ or at infinity.
\item $\sigma_{ess}(D_A^\varepsilon)=(-\infty,-|m|]\cup[|m|,+\infty)$.
\end{enumerate}
\end{proposition}

\begin{proof}
$a)$ First, note that if $\eta_\varepsilon$ is constant then $\eta_\varepsilon\equiv 1$, since $\lim_{a\to +\infty}\eta_\varepsilon(\pm ia)=1$,
because $\lim_{a\to +\infty}\alpha_1(\varepsilon k(\pm ia))=0$,
by the dominated convergence theorem. Consequently, $\eta_\varepsilon$ has no zeros and there are no eigenvalues. (In the case $m=0$, one concludes the same for each connected component of the domain of $\eta_\varepsilon$, i.e., $\C_\pm$.) Now, assume that $\eta_\varepsilon$ is non-constant. Then  $a)$ follows immediately from Proposition \ref{prop:ev}, Lemma \ref{lemma:analyticity}, and the identity theorem for analytic functions.

$b)$ Since the resolvent of $D_A^\varepsilon$ is a compact perturbation of the free resolvent $R_z$ and $\res(D_A^\varepsilon)\cap\C_\pm\neq\emptyset$, by $a)$ and Proposition \ref{prop:approx_res}, we may just apply the Weyl essential spectrum theorem \cite[Thm. XIII.14]{RS4}.

\end{proof}

\begin{remark}
One can  localize the eigenvalues of $D_A^\varepsilon$ using spectral enclosures derived in \cite[Thm. 2.1, Cor. 2.6]{CuTr_16}. In particular, for every $z\in\sigma(D_A^\varepsilon)$, we have $|\Im z|\leq\|A\otimes W_\varepsilon\|=\|A\|\varepsilon^{-1}\|v\|_{L^2(\R)}^2$. The upper bound clearly explodes as $\varepsilon\to 0$. Another enclosure is given by the union of a ball centred at the origin and sectors lying symmetrically around the real axis. However, this enclosure can not be made uniform in $\varepsilon$ either.
\end{remark}

At the end of this section, we will show with an example that in the \textit{strongly non-self-adjoint} case the point spectrum of the formal limiting operator may suddenly expand. Put
$$A=\begin{pmatrix}
0 & 2\\
-2 & 0
\end{pmatrix}.$$
Then, according to Theorem \ref{theo:spec}, $\sigma_p(D_A)=\C\setminus((-\infty,-|m|]\cup[|m|,+\infty))$. Next, we have
$$\det\Big(\sigma_0+\frac{i}{2} \alpha_1(\varepsilon k(z)) A Z(z) \Big)=1-\alpha_1(\varepsilon k(z))^2.$$
Hence, \eqref{eq:ev_DAeps} is equivalent to 
\begin{equation} \label{eq:evEx}
\alpha_1(\varepsilon k(z))^{2}=1.
\end{equation}
Note that $w\mapsto \alpha_1(w)$ is non-constant on $\C_+$, because
$$\lim_{w\to 0,\Im w>0}\alpha_1(w)=1 \quad\text{and}\quad \lim_{a\to+\infty}\alpha_1(ia)=0.$$
Assume that $v$ is compactly supported. Then $w\mapsto \alpha_1(w)$ is well defined and analytic on $\C$. By the identity theorem, the equation 
\begin{equation} \label{eq:evEx2}
\alpha_1(w)^2=1
\end{equation}
has at most countably many zeros that may accumulate only  at infinity. Moreover, if $w\in\C_+$ satisfies \eqref{eq:evEx2} then the corresponding eigenvalues of $D_A^\varepsilon$ are given by 
$$z_\pm=\pm\sqrt{\left(\frac{w}{\varepsilon}\right)^2+m^2}.$$
We conclude that for a given compact set $\mathcal{C}$ in $\C\setminus((-\infty,-|m|]\cup[|m|,+\infty))$, there exists $\varepsilon_{\mathcal{C}}$ such that for all $\varepsilon\in(0,\varepsilon_{\mathcal{C}})$, the set $\mathcal{C}$ contains no eigenvalues of $D_A^\varepsilon$.

\section{Non-relativistic limit} \label{sec:limit}

\subsection{Introducing the speed of light}
In the previous sections, we put the speed of light $c$ equal to one to have a concise notation. However, it is necessary  to introduce $c$-dependent quantities before performing the non-relativistic limit. To this purpose we start with the following formal differential expression
$$\mathscr{D}^{m,c}=-c\sigma_1(-i\dd/\dd x)+\sigma_3 mc^2+cA\delta$$
and denote by $D_A^{m,c}$ the operator associated to this expression by the same procedure as in the case $c=1$. Then $D_A=D_A^{m,1}$ and the transmission condition for  the functions in $\dom(D_A^{m,c})$ is still given by \eqref{eq:TC}, i.e., it is independent of $m$ and $c$. Consequently, we get
$$(D_A^{m,c}-z)^{-1}=(cD_A^{mc,1}-z)^{-1}=\frac{1}{c}\Big(D_A^{mc,1}-\frac{z}{c}\Big)^{-1}.$$
Using this scaling property together with \eqref{eq:resolvent} and \eqref{eq:freeResolvent} we deduce that $R_z^{A,c}:=(D_A^{m,c}-z)^{-1}$ is the integral operator with the kernel
\begin{equation*}
R_z^{A,c}(x,y)=R_z^c(x,y)-cR_z^c(x,0)\Big(\sigma_0+\frac{i}{2}AZ_c(z)\Big)^{-1}A R_z^c(0,y),
\end{equation*}
where 
\begin{equation*}
R_z^c(x,y):=\frac{i}{2c}(Z_c(z)+\sgn(x-y)\sigma_1)\ee^{ik_c(z)|x-y|}
\end{equation*}
is the integral kernel of the free resolvent, i.e., $R_z^c=(D_0^{m,c}-z)^{-1}$ and
$$Z_c(z):=\begin{pmatrix}
\zeta_c(z) & 0\\
0 & \zeta_c^{-1}(z)
\end{pmatrix},
\quad \zeta_c(z):=\frac{z+mc^2}{ck(z)},\quad ck_c(z):=\sqrt{z^{2}-(mc^2)^{2}}.$$

\subsection{General non-relativistic point interaction} From now on, let $m>0$. By the non-relativistic limit of $D_A^{m,c}$ we mean the limit $\lim_{c\to+\infty}(D_A^{m,c}-mc^2)$, if it exists in some sense. In the self-adjoint setting $(A=A^*)$, it was proved that such a limit exits in the "norm-resolvent topology" after making the transmission condition $c$-dependent \cite{BeDa_94}. 
This extends the classical result for the free Dirac operator which says that
\begin{equation} \label{eq:free_limit}
\lim_{c\to+\infty}(D_0^{m,c}-mc^2-z)^{-1}=\begin{pmatrix}1 & 0\\ 0 & 0\end{pmatrix}\otimes(H_0-z)^{-1}
\end{equation}
in the operator norm, cf. \cite{GeGrTh_84}. Here, 
\begin{equation*}
H_0=-\frac{1}{2m}\frac{\dd^2}{\dd x^2} \quad\text{with } \dom(H_0)=H^2(\R)
\end{equation*}
stands for the  non-relativistic one-dimensional free Hamiltonian. After restricting $H_0$ to the functions that vanish at $x=0$ together with their first derivatives, it is possible to construct a four-real-parametric family of self-adjoint extensions \cite{Se_86gen}, which are reffered to as the generalized point interactions. In \cite{BeDa_94}, the transmission condition was chosen $c$-dependent in the way that one gets all of these extensions as the non-relativistic limits of their relativistic counterparts. 
Note that a different parametrization of the transmission conditions that does not make use of the matrix $A$ was employed there. However, one can deduce that in our setting elements of $A$ should be scaled as
\begin{equation} \label{eq:A_c}
A\mapsto A_c=\begin{pmatrix}
\frac{1}{2mc}\alpha & \beta\\
\gamma & 2mc\delta
\end{pmatrix}
\end{equation}
when performing the non-relativistic limit. Otherwise, it may happen that the limit does not exist or it is just the free operator.

Our aim  is to reproduce the whole family of non-relativistic not necessarily self-adjoint generalized point interactions as the non-relativistic limits of their relativistic counterparts. To introduce such a family we start with formal differential expressions
\begin{equation*}
\mathscr{H}:=-\frac{1}{2m}\frac{\dd^2}{\dd x^2},\quad \mathscr{H}_A:=\mathscr{H}+\frac{1}{2m}\big(\alpha|\delta_0\rangle\langle\delta_0|+i\beta|\delta_0\rangle\langle\delta'_0|-i\gamma|\delta'_0\rangle\langle\delta_0|+\delta|\delta'_0\rangle\langle\delta'_0|\big),
\end{equation*}
where $m>0$ and $\alpha,\,\beta,\,\gamma,\,\delta\in\C$ are the entries of the matrix $A$, see  \eqref{eq:A}. By the inner product of the Dirac distribution  $\delta_0$ or its derivative $\delta_0'$ with a smooth function $\psi$ we mean just the distributional action $\langle\delta_0,\psi\rangle=\psi(0)$ or $\langle\delta_0',\psi\rangle=-\psi'(0)$. Now, let $\psi\in H^2(\R_-)\oplus H^2(\R_+)$. Then we extend the actions of these distributions as follows,
\begin{equation*}
\langle\delta_0,\psi\rangle:=\frac{\psi(0_+)+\psi(0_-)}{2},\quad \langle\delta_0',\psi\rangle:=-\frac{\psi'(0_+)+\psi'(0_-)}{2},
\end{equation*}
cf. \cite{Ku_96} or \cite[Sect. 3.2.4]{singular}. If we want $\mathscr{H}_A\psi$ to belong to $L^2(\R)$ then the singular contributions have to cancel out which yields
\begin{align*}
 - (\psi'(0+)-\psi'(0-))+ \alpha \frac{\psi(0+)+\psi(0-)}{2}-i\beta\frac{\psi'(0+)+\psi'(0-)}{2}  &= 0,\\
-(\psi(0+)-\psi(0-))- i\gamma\frac{\psi(0+)+\psi(0-)}{2} - \delta\frac{\psi'(0+)+\psi'(0-)}{2}&= 0.
\end{align*}
Introducing the matrix
\begin{equation} \label{eq:Vdef}
V:=\begin{pmatrix} i & 0\\ 0 & 1\end{pmatrix}
\end{equation}
and the boundary value operators $\tilde\Gamma_{j}:\, H^2(\R_-)\oplus H^2(\R_+)\to \C^2,\, j=1,2$,
\begin{equation} \label{eq:BT}
\tilde\Gamma_1\psi:=\begin{pmatrix}
\psi'(0+)-\psi'(0-)\\
\psi(0+)-\psi(0-)
\end{pmatrix},\quad
\tilde\Gamma_2\psi:=\frac{1}{2}\begin{pmatrix}
\psi(0+)+\psi(0-)\\
-\psi'(0+)-\psi'(0-)
\end{pmatrix},
\end{equation}
this can be rewritten as
\begin{equation} \label{eq:TC_nonrel}
\tilde\Gamma_1\psi=VAV^*\tilde\Gamma_2\psi.
\end{equation}
Therefore, we are motivated to define the non-relativistic Hamiltonian with general (not necessarily self-adjoint) point interaction as
\begin{equation} \label{eq:H_A}
\begin{split}
&\dom(H_A)=\{\psi\equiv\psi_-\oplus\psi_+\in H^2(\R_-)\oplus H^2(\R_+)|\, \text{\eqref{eq:TC_nonrel} holds}\}\\
&H_A\psi=\mathcal{H}\psi_-\oplus\mathcal{H}\psi_+.
\end{split}
\end{equation}

The amount of literature on operators $H_A$ in the self-adjoint setting , i.e., if and only if $A=A^*$ (as we will show in Proposition \ref{prop:app1}), is so huge that we will not even try to provide an overview on the subject. Nevertheless, one of the pilot articles \cite{Se_86gen} and the monograph \cite{solvable} may serve as good starting points for a practically endless chain of references. Interestingly, a similar formal expression as $\mathscr{H}_A$ was considered in \cite{Se_86gen} even in the non-self-adjoint setting and the corresponding Green's function was derived there in a heuristic manner. Later,  the questions of  $\mathcal{P}\mathcal{T}$-symmetry of $H_A$ and  similarity of $H_A$ to a self-adjoint operator were addressed in \cite{GrKu_14}. In the present paper, we are using practically the same transmission condition as was deduced there.  Similar problems as in \cite{GrKu_14} were studied in \cite{HuKrSi_15} in a much more general setting of a quantum graph. Of course, a different parametrization of the transmission condition, that suits better the graph problems, was used there. Although in all of three mentioned papers \cite{Se_86gen,GrKu_14,HuKrSi_15}, a certain formula for the resolvent is included, none can be used directly for our purposes. For this reason and also for the reader's convenience, we decided to include a short calculation of the resolvent of $H_A$ in Appendix along with other basic results for $H_A$. We will show that $(H_A-z)^{-1}$ is the integral operator with the kernel
\begin{equation} \label{eq:non_rel_res}
(H_A-z)^{-1}(x,y)=(H_0-z)^{-1}(x,y)-\Big(\frac{1}{i\mu(z)}f_z(x),g_z(x)\Big)K_A(z)
\begin{pmatrix}
\frac{1}{i\mu(z)}f_z(y)\\
g_z(y)
\end{pmatrix},
\end{equation}
where
\begin{equation} \label{eq:basis}
\mu(z):=\sqrt{2m z},\quad f_z(x):=\ee^{i\mu(z)|x|},\quad g_z(x):=\sgn(x)f_z(x),
\end{equation}
and
\begin{equation} \label{eq:K_A}
K_A(z):=\frac{i m}{4-\det(A)+\frac{2i}{\mu(z)}\alpha+2i\mu(z)\delta}
\begin{pmatrix}
\mu(z)\det(A)-2i\alpha & 2\beta\\
-2\gamma & \frac{\det(A)}{\mu(z)}-2i\delta
\end{pmatrix}.
\end{equation}
For the sake of completeness, let us also recall that
\begin{equation*}
(H_0-z)^{-1}(x,y)=\frac{im}{\mu(z)}\ee^{i\mu(z)|x-y|}.
\end{equation*} 

\begin{theorem} \label{theo:limit}
Let $m>0$ and $z\in\C\setminus[0,+\infty)$ be such that 
\begin{equation} \label{eq:det_ass}
4-\det(A)+\frac{2i}{\mu(z)}\alpha+2i\mu(z)\delta\neq 0,
\end{equation}
and $A_c$ be as in \eqref{eq:A_c}. Then for all sufficiently large $c$, $z\in\res(D^{m,c}_{A_c}- mc^2)$ and
\begin{equation*}
\lim_{c\to+\infty}\big\|(D^{m,c}_{A_c}-mc^2-z)^{-1}-\begin{pmatrix}1 & 0\\ 0 & 0\end{pmatrix}\otimes (H_A-z)^{-1}\big\|=0.
\end{equation*}
\end{theorem}

\begin{proof}
Using the fact that $(D^{m,c}_{A_c}-mc^2-z)^{-1}=R^{A_c,c}_{z+mc^2}$ together with the resolvent formulae \eqref{eq:resolvent} and \eqref{eq:non_rel_res}, and the convergence result  for the free operators \eqref{eq:free_limit}, we conclude that it is enough to show the uniform convergence of the integral operator with the kernel
\begin{equation} \label{eq:res_diff_rel}
cR_{z+mc^2}^c(x,0)\Big(\sigma_0+\frac{i}{2}A_cZ_c(z+mc^2)\Big)^{-1}A_c R_{z+mc^2}^c(0,y)
\end{equation}
to the integral operator with the kernel
\begin{equation} \label{eq:res_diff_nonrel}
\begin{pmatrix}1 & 0\\ 0 & 0\end{pmatrix}\otimes\Big(\frac{1}{i\mu(z)}f_z(x),g_z(x)\Big)K_A(z)
\begin{pmatrix}
\frac{1}{i\mu(z)}f_z(y)\\
g_z(y)
\end{pmatrix}
\end{equation}
as $c\to+\infty$. Since both operators are Hilbert-Schmidt, it is sufficient to prove the convergence of the matrix elements of the kernels in $L^2(\R^2,\dd x\dd y)$.

We have
\begin{align}
& k_c(z+mc^2)=\sqrt{\left(\frac{z}{c}\right)^2+2mz}=\mu(z)+\mathcal{O}
(c^{-2}), \label{eq:k_asy}\\
& \zeta_c(z+mc^2)=\frac{z+2mc^2}{ck_c(z+mc^2)}=\frac{2mc}{\mu(z)}+\mathcal{O}(c^{-1}), \nonumber\\
& \zeta_c(z+mc^2)^{-1}=\frac{\mu(z)}{2mc}+\mathcal{O}(c^{-3}) \nonumber,
\end{align}
as $c\to+\infty$, and $\det(A_c)=\det(A)$ for all $c\neq 0$. Hence, we get
$$\det\Big(\sigma_0+\frac{i}{2}A_cZ_c(z+mc^2)\Big)=\frac{1}{4}\Big(4-\det(A)+\frac{2i}{\mu(z)}\alpha+2i\mu(z)\delta\Big)+\mathcal{O}(c^{-2}).$$
Due to \eqref{eq:det_ass}, the right-hand side is non-zero for every sufficiently large $c$. Moreover, since $D_0^{m,c}=c D_0$, $\res(D_0^{m,c}-mc^2)=\C\setminus((-\infty,-2mc^2]\cup[0,+\infty))$. Using an obvious variation of Theorem \ref{theo:res}, we infer that for every $c$ above a $z$-dependent threshold, $z\in\res(D_{A_c}^{m,c}-mc^2)$. Next, after a tedious but straightforward calculation, we deduce that
\begin{multline*}
cR_{z+mc^2}^c(x,0)\Big(\sigma_0+\frac{i}{2}A_cZ_c(z+mc^2)\Big)^{-1}A_c R_{z+mc^2}^c(0,y)\\
=\begin{pmatrix}1+\mathcal{O}(c^{-1}) & \mathcal{O}(c^{-1})\\ \mathcal{O}(c^{-1})  & \mathcal{O}(c^{-1})\end{pmatrix}\otimes\Big(\frac{1}{i\mu(z)}f_z^c(x),\sgn(x)f_z^c(x)\Big)K_A(z)
\begin{pmatrix}
\frac{1}{i\mu(z)}f_z^c(y)\\
\sgn(y)f_z^c(y)
\end{pmatrix}
\end{multline*}
as $c\to+\infty$, where
$$f_z^c(x):=\ee^{ik_c(z+mc^2)|x|}.$$

Taking \eqref{eq:k_asy} into account, we see that the point-wise limit of \eqref{eq:res_diff_rel} is \eqref{eq:res_diff_nonrel}.
To establish the $L^2$-convergence, it is sufficient to show that
\begin{equation} \label{eq:f_conv}
\lim_{c\to+\infty}\|f_z^c-f_z\|_{L^2(\R)}=0.
\end{equation}
Recall that with our convention, $\Im\sqrt{w}>0$ for all $w\in\C\setminus[0,+\infty)$. By continuity of the complex square root, we see that for every $c$ above another $z$-dependent threshold, 
$$0<\frac{1}{2}\Im\mu(z)<\Im\sqrt{k_c(z+mc^2)}.$$
Therefore, we get $|f_z^c(x)|\leq\ee^{-\Im\mu(z)|x|/2}$ for all $x\in\R$. Using the dominated convergence theorem, we arrive at \eqref{eq:f_conv}.
\end{proof}

\section{Appendix}
In this appendix we review several basic facts about the operator $H_A$ defined by \eqref{eq:H_A}. Following \cite{AlNi_07}, we see that
$$H_{\min}\subset H_A\subset H_{\max}=(H_{\min})^*,$$
where 
\begin{align*}
&\dom(H_{\min})=\{\psi\in H^2(\R)|\, \psi(0)=\psi'(0)=0\},\quad H_{\min}\psi=-\frac{1}{2m}\psi''\\
&\dom(H_{\max})=H^2(\R_-)\oplus H^2(\R_+), \quad H_{\max}(\psi_-\oplus\psi_+)=-\frac{1}{2m}\psi_-''\oplus -\frac{1}{2m}\psi_+'',
\end{align*}
and that for all $\psi,\phi\in\dom(H_{\max})$,
\begin{equation*}
\langle\psi, H_{\max}\phi\rangle-\langle H_{\max}\psi,\phi\rangle=\frac{1}{2m}\big(\langle\tilde\Gamma_2\psi,
\tilde\Gamma_1\phi\rangle_{\C^2}-\langle\tilde\Gamma_1\psi,
\tilde\Gamma_2\phi\rangle_{\C^2}\big).
\end{equation*}
The boundary value mappings $\tilde\Gamma_j,\, j=1,2,$ were introduced in \eqref{eq:BT}. Since $\psi\mapsto(\tilde\Gamma_1,\tilde\Gamma_2)$ is a surjective map from $\dom(H_{\max})$ onto $\C^2\times\C^2$, the triplet 
$$(\C^2,\frac{1}{\sqrt{2m}}\tilde\Gamma_1,-\frac{1}{\sqrt{2m}}\tilde\Gamma_2)=:(\C^2,\hat\Gamma_1,\hat\Gamma_2)$$
is a boundary triplet for $H_{\max}$. We have
\begin{align}
\dom(H_A)&=\{\psi\in\dom(H_{\max})|\, \hat\Gamma_1\psi=-VAV^*\hat\Gamma_2\psi\} \label{eq:domH_A}\\
&=\{\psi\in\dom(H_{\max})|\, (\hat\Gamma_1\psi,\hat\Gamma_2\psi)\in\Theta_A\}, \nonumber
\end{align}
where  $V$ is given by \eqref{eq:Vdef} and 
$$\Theta_A:=\{(-VAV^*\xi,\xi)|\, \xi\in\C^2\}\subset \C^2\times \C^2$$ 
is a closed linear relation in $\C^2$. Note that
\begin{equation*}
\dom(H_0)=\{\psi\in\dom(H_{\max})|\,\hat\Gamma_1\psi=0\}=H^2(\R),
\end{equation*}
i.e., $H_0$ is non-relativistic free Hamiltonian, as expected. It is a well known fact that the spectrum of $H_0$ is purely absolutely continuous and equal to $[0,+\infty)$.

Using the boundary triplets techniques, we will prove the following statements.
\begin{proposition} \label{prop:app1}
The operator $H_A$ is closed  and it is self-adjoint if and only if $A$ is hermitian.
\end{proposition}

\begin{proof}
The first statement follows immediately from \cite[Lem. 14.6 (ii)]{Schmudgen}

To prove the second statement we may use  \cite[Thm. 1.12 and Cor. 1.6]{BrGePa_08}, which says that $H_{\max}$ restricted to the functions that obey $C\hat\Gamma_1\psi=D\hat\Gamma_2\psi$ is self-adjoint if and only if
\begin{equation} \label{eq:saCond}
CD^*\text{ is hermitian}\quad \wedge \quad \det(CC^*+DD^*)\neq 0.
\end{equation}
In view of \eqref{eq:domH_A}, this means that if $H_A$ is self-adjoint then $CD^*=-V A^*V^*$ is hermitian, which implies that $A$ itself has to be hermitian.
On the other hand,  if $A$ is hermitian then the same holds true for $CD^*$. Moreover, since $CC^*+DD^*=\sigma_0+V AA^* V^*$, the second condition in \eqref{eq:saCond} is equivalent to $-1\notin\sigma(V AA^* V^*)$, which is always true, because $V AA^* V^*$ represents a positive operator. We conclude that $H_A$ is self-adjoint. 
\end{proof}

\begin{proposition}
Let $z\in\res(H_0)=\C\setminus[0,+\infty)$, $m>0$ and $\mu(z)=\sqrt{2m z}$. Then $z\in\res(H_A)$ if and only if
\begin{equation} \label{eq:res_cond}
4-\det(A)+\frac{2i}{\mu(z)}\alpha+2i\mu(z)\delta\neq 0.
\end{equation}
In the positive case, the resolvent of $H_A$ at $z$ is given by \eqref{eq:non_rel_res}.
\end{proposition}

\begin{proof}
Since the relation $\Theta_A$ that describes the domain of $H_A$ has a tight parametric representation, we can use Corollary 2.6.3 of \cite{Snoo}, which says that for $z\in\res(H_0)$,
\begin{equation} \label{eq:Krein_formula}
(H_A-z)^{-1}-(H_0-z)^{-1}=\hat\gamma(z)\mathcal{A}(\mathcal{B}-\hat M(z)\mathcal{A})^{-1}\hat\gamma(\bar z)^*,
\end{equation}
whenever $(\mathcal{B}-\hat M(z)\mathcal{A})^{-1}$ exists as a linear operator on $\C^2$.  Here,
$$\mathcal{A}=-VAV^*,\quad \mathcal{B}=\sigma_0,\quad \hat\gamma(z)=(\hat\Gamma_1|_{\ker(H_{\max}-z)})^{-1},\quad \hat M(z)=\hat\Gamma_2\circ\hat\gamma(z).$$
Choosing $\{f_z,g_z\}$ given in \eqref{eq:basis} for the basis of $\ker(H_{\max}-z)$, it is straightforward to show that
\begin{equation} \label{eq:gammaHat}
\hat\gamma(z)\begin{pmatrix} a\\b\end{pmatrix}=\frac{\sqrt{2m}}{2}\Big(\frac{a}{i\mu(z)}f_z+b g_z\Big)\quad (\forall a,b\in\C)
\end{equation}
and  
\begin{equation}
\hat{M}(z)=\frac{i}{2}\begin{pmatrix}
\frac{1}{\mu(z)} & 0\\
0 & \mu(z)
\end{pmatrix}.
\end{equation}
Now, the matrix $(\mathcal{B}-\hat M(z)\mathcal{A})$ is invertible if and only if \eqref{eq:res_cond} is satisfied and in the positive case,
\begin{equation} \label{eq:K_A_mult}
\mathcal{A}(\mathcal{B}-\hat M(z)\mathcal{A})^{-1}=-\frac{2}{m}K_A(z)
\end{equation}
with $K_A(z)$ given by \eqref{eq:K_A}.
Finally, using the observation that with our choice of the branch of the complex square-root $\mu(\bar z)=-\overline{\mu(z)}$,  we get
\begin{equation} \label{eq:gammHatAdj}
\hat\gamma(\bar z)^*\phi=\frac{\sqrt{2m}}{2}
\begin{pmatrix}
\frac{1}{i\mu(z)}\int_\R f_z(x)\phi(x)\dd x\\
\int_\R g_z(x)\phi(x)\dd x
\end{pmatrix}.
\end{equation}
Plugging \eqref{eq:gammaHat}, \eqref{eq:K_A_mult}, and \eqref{eq:gammHatAdj} into \eqref{eq:Krein_formula}, we arrive at \eqref{eq:non_rel_res}.
\end{proof}

\begin{remark}
Either by a direct calculation or employing \cite[Thm. 2.6.2]{Snoo}, one concludes that the eigenvalues of $H_A$ in $\C\setminus[0,+\infty)$ are exactly the solutions to 
\begin{equation*}
4-\det(A)+\frac{2i}{\mu(z)}\alpha+2i\mu(z)\delta=0,
\end{equation*}
cf. \cite[Lem. 2.3]{GrKu_14}.
\end{remark}

\section*{Acknowledgments}
L. Heriban acknowledges  the support by the EXPRO grant No. 20-17749X of the Czech Science Foundation (GA\v{C}R). M.~Tu\v{s}ek was partially supported by the grant No.~21-07129S of the Czech Science Foundation (GA\v{C}R) and by the project CZ.02.1.01/0.0/0.0/16\_019/0000778 from the European Regional Development Fund.

\end{document}